\documentclass[a4paper]{article}
\pdfoutput=1

\usepackage{amsmath}
\usepackage{amssymb}
\usepackage{amsthm}
\usepackage{amsfonts}

\usepackage{newpxtext}
\usepackage{newpxmath}

\usepackage[margin=2.5cm]{geometry}

\usepackage{algorithm}
\usepackage{algpseudocode}

\usepackage[hyphens]{url}
\usepackage[hidelinks]{hyperref}

\usepackage{graphicx}

\usepackage{thm-restate}

\newtheorem{theorem}{Theorem}
\newtheorem{lemma}[theorem]{Lemma}
\newtheorem{corollary}[theorem]{Corollary}
\newtheorem{proposition}[theorem]{Proposition}
\newtheorem{definition}[theorem]{Definition}
\newtheorem{remark}[theorem]{Remark}

\DeclareMathOperator{\E}{E}
\DeclareMathOperator{\Var}{Var}

\title{Huffman-Bucket Sketch \\
    \small A Simple $O(m)$ Algorithm for Cardinality Estimation
}
\author{Matti Karppa \\
    \small University of Gothenburg, Sweden \\
    \small Chalmers University of Technology, Sweden
}
\date{}

\begin{document}

\maketitle

\begin{abstract}
\noindent We introduce the \emph{Huffman-Bucket Sketch (HBS)}, a simple, mergeable data structure that losslessly compresses a 
HyperLogLog (HLL) sketch with $m$ registers to optimal space $O(m+\log n)$ bits, with amortized constant-time updates,
acting as a drop-in replacement for HLL that retains mergeability and substantially reduces memory requirements. 
We partition registers into small buckets and encode their values with a global Huffman codebook derived from 
the strongly concentrated HLL rank distribution, using the current cardinality estimate for determining the mode
of the distribution.
We prove that the Huffman tree needs rebuilding only $O(\log n)$ times over a stream, roughly when cardinality doubles.
The framework can be extended to other sketches with similar strongly concentrated
distributions. We provide preliminary numerical evidence that suggests that HBS is practical and can potentially be competitive
with state-of-the-art in practice.
\end{abstract}

\section{Introduction}
\label{sec:introduction}

Estimating the \emph{number of distinct elements}, or \emph{cardinality} in a massive data stream is a standard primitive in
databases~\cite{Trummer:2019,ChiaDPSLDWG:2019,FreitagN:2019,WangCLL:2021,PavlopoulouCT:2023}, 
networking~\cite{ChabchoubCD:2014,Ben-BasatEFMR:2018,DaoJJMN:2022,ClemensSGH:2023}, and metagenomics~\cite{BakerLangmead:2019,MarcaisSPK:2019},
among others. As the universe of the elements can be essentially infinite, such as the set of all variable-length strings,
storing a representative for each element takes $\Omega(n\log n)$ bits, by the pigeonhole principle, 
assuming the true cardinality of the stream is~$n$.
Even if we can initialize binary counters for all elements in the universe, an exact count needs
$\Theta(n)$ bits~\cite{AlonMS:1999}, which can be infeasibly large. 
However, the cardinality can be estimated with fewer bits, 
and it is known that $\Theta(m)$ bits are necessary and sufficient to achieve 
a relative standard error of $\sqrt{\Var\left[\frac{\hat{n}}{n}\right]}=O(1/\sqrt{m})$~\cite{KaneNW:2010}.

A practical standard solution to the distinct elements problem is the HyperLogLog (HLL) sketch \cite{FlajoletFGM:2007},
which has size $O(m\log\log n)$ bits for a relative error of $O(1/\sqrt{m})$. 
It is attractive because it is almost trivial to implement, can be implemented incredibly efficiently,
has constant-time updates, and is mergeable, meaning that a large datastream can easily be processed in 
a distributed fashion by maintaining sketches of substreams that can later be merged into one sketch, as if
the entire stream had been processed by a single sketch.

A lot of the work on improving the space usage over HLL has had to compromise some of these properties,
such as mergeability~\cite{PettieWY:2021} or truly constant-time updates~\cite{KarppaP:2022}.
In this paper, we present a simple algorithm and an accompanying data structure, the \emph{Huffman-Bucket Sketch (HBS)}, 
that 
can \emph{losslessly compress} the HLL sketch to
$O(m)$ bits using Huffman codes, while preserving mergeability and sustaining efficient updates. 
While not truly constant time,
the updates are amortized $O(1)$, and under mild assumptions on the relationship between $n$ and $m$,
\emph{most} updates can be constant time, with only a small fraction of updates requiring more time to update the codewords.

The main idea behind the algorithm is to leverage the fact that the distribution of register values (or \emph{ranks})
in the HLL sketch is highly concentrated around $\lceil \log_2(n/m)\rceil$, with very fast decaying tails.
In fact, as it is well known, the entropy of the register values is asymptotically constant per register~\cite{Durand:2004},
and this manifests in the fact that if we split the registers into $O(\log n)$-sized buckets,
the distribution of the Huffman codeword lengths in each bucket is also highly concentrated, requiring
only $O(\log n)$ bits per bucket with high probability.

We show that, as we process a stream of $n$ distinct elements, the Huffman tree needs to be reconstructed only $O(\log n)$
times, essentially only when the cardinality is doubled, which we can detect by maintaining a local cardinality estimate.
That is, even though the true cardinality~$n$ is unknown to us, we can pull ourselves up from the swamp by our
hair, much like Baron von M\"unchhausen,
by using the current global estimate to approximate the register value distribution, which is unambiguously determined
by the cardinality for a fixed number of registers, and subsequently unambiguously determines the Huffman tree.

The resulting data structure supports the regular operations of estmating a cardinality,
inserting a new element in the sketch, querying a register value, and merging two sketches.
The asymptotic size of the sketch is $O(m+\log n)$ bits, which is optimal~\cite{KaneNW:2010}, 
and the amortized time for insertions is $O(1)$, while merges can be carried out in $\tilde O(m)$ time.\footnote{$\tilde O(\cdot )$ suppresses polylogarithmic factors.}

The idea of bucketing originates from practical considerations: if the size of the bucket is sufficiently small,
such as a single machine word or a cache line, then any operations performed on it are essentially constant time.
We complement the theoretical, asymptotic analysis by considering the practical implications and implementation
concerns of the algorithm, and show numerical evidence that suggests that the algorithm is competetive with 
the state-of-the-art, especially accounting the fact that it is, in fact, a drop-in replacement for HLL
that can be decompressed back to a HLL sketch at any point, and thus maintains the same mergeability properties.
We do not analyze the estimation procedure in this paper,
as any HLL-compatible estimator can be used, such as 
the original HLL estimator, Ertl's MLE estimator~\cite{Ertl:2017}, or Pettie and Wang's 
state-of-the-art estimator~\cite{PettieW:2021} with relative variance of $1.075/m$.

The rest of the paper is organized as follows.
In Section~\ref{sec:preliminaries}, we review the notation, the mathematical preliminaries, the Poisson model, and recap HLL.
Section~\ref{sec:algorithm} is the main result
section, describing the HBS algorithm and the data structure.
Section~\ref{sec:analysis} provides asymptotic analysis of the algorithm in the Poissonized balls-and-bins model. 
In Section~\ref{sec:generalizations}, we briefly outline extensions of the framework. Section~\ref{sec:practicalimplications}
discusses implementation considerations and empirical implications, including the Memory-Variance Product (MVP) of the sketch.

\subsection{Related work}

The line of work leading to HLL and its variants was initiated by Flajolet and Martin with their
\emph{probabilistic counting} algorithm~\cite{FlajoletM:1985}, the foundation of which was a sparse indicator 
matrix, the \emph{FM85 matrix}, 
whose columns correspond to geometrically distributed \emph{ranks} of hash values, and rows to uniformly
distributed substreams. This was refined by Durand and Flajolet's LogLog algorithm~\cite{DurandF:2003}, 
whose key innovation was compressing the FM85 matrix by only retaining the maximum rank per row.
HLL~\cite{FlajoletFGM:2007} was obtained from this by replacing the arithmetic mean in the
estimator with the harmonic mean, improving the accuracy and robustness of the estimator, and has
since become the standard solution for cardinality estimation in practice.

HLL has seen a large number of refinements and variants, both practical and theoretical, including 
HyperLogLog++\cite{HeuleNH:2013}, HLL-Tailcut~\cite{XiaoCZL:2020}, and HyperLogLogLog~\cite{KarppaP:2022}.
Another, related and simple cardinality estimation method is linear counting~\cite{WhangVT:1990}, 
which can be seen as a Bloom filter with one hash function, and is optimal for small cardinalities, assuming
$m = \Theta(n)$, and is used as small-cardinality correction in HLL.

Kane, Nelson, and Woodruff~\cite{KaneNW:2010}, later improved by B\l{}asiok~\cite{Blasiok:2020}, 
showed the existence of an optimal $O(m+\log n)$-bit sketch with constant time
updates, but this sketch is widely considered impractical.
Pettie and Wang have established sharp information-theoretic limits~\cite{PettieW:2021} for cardinality estimation,
and observed that the HLL sketch loses some information of the FM85 matrix that would be beneficial
for further compression, confirming the observation by Lang~\cite{Lang:2017} that the FM85 matrix can be compressed
beyond the entropy of the HLL sketch. This has led to various other work that extends
the compression using additional information from the FM85 matrix, including
UltraLogLog~\cite{Ertl:2024}, ExaLogLog~\cite{Ertl:2025}, and SpikeSketch~\cite{DuHSLZG:2023}.
Another line of work has focused on forfeiting mergeability by using the martingale transform~\cite{PettieWY:2021}.
Finally, while we only address compression of the sketch, a lot of the other work has also focused on improving
the estimators~\cite{Ertl:2017,Ertl:2024,WangP:2023}.

\section{Preliminaries}
\label{sec:preliminaries}

\subsection{Mathematical notation and preliminaries}
We write $[n]=\{1,2,\ldots,n\}$ for the set of the first $n$ positive integers, and 
the Iverson bracket notation $\llbracket P \rrbracket$ to denote the indicator of a predicate~$P$, 
which is 1 if $P$ is true and 0 otherwise.

\subsection{Poissonized balls-and-bins model}
The analysis of the algorithm is carried out in the Poissonized balls-and-bins model,
which is a standard model for analyzing such 
algorithms~\cite{Ertl:2017,Ertl:2024,Ertl:2025,PettieW:2021,PettieWY:2021,WangP:2023}, 
and is known to be almost equivalent to the exact model.

In the balls-and-the-bins model, we throw $n$ balls independently and uniformly at random into $m$ bins, and analyze 
the distribution. The downside of this model is that the bins are dependent. To address this, we Poissonize the model
by letting balls arrive independently at random in the bins, with rate $\lambda=\frac{n}{m}$,
meaning that the number of balls in each bin is independently and identically distributed as $\mathrm{Pois}(\lambda)$. 
This makes $n$ a random variable,
but assuming the number of bins is sufficiently large, the distribution is well known to be almost equal.\footnote{For a standard
textbook reference, see~\cite{MitzenmacherU:2005}.}
In the actual analysis, we have two conjoined processes where the row of the FM85 matrix is selected uniformly at random
but the column is selected according to the geometric distribution. Nevertheless, the same Poissonization technique applies.

\subsection{HyperLogLog (HLL)}
Let $w$ be the word size. We assume that all hash functions $h : \mathcal U \to [2^w]$ are fully random and yield
$w$-bit hashes. As we are going to hash all elements, by the pigeonhole principle,
we require that $w=\Omega(\log n)$ to be able to tell all elements apart.

The HLL sketch~\cite{FlajoletFGM:2007} consists of an array $M$ of $m$ registers, each of which records the maximum \emph{rank} of hashes of
substreams that are mapped to the register. All registers are initialized to zero. Upon receiving a new element~$y$,
the element is hashed into a $w$-bit value $x=h(y)$. The hash value is split in two parts:
$\log_2 m$ bits of $x$ are used to select the register index~$j\in[m]$, and the remaining bits are used 
to compute the rank $r=\rho(x)$, which is the 1-based position of the leftmost 1-bit in the remaining part of the hash value.
It is immediate that $r\sim \mathrm{Geom}(1/2)$. The update rule is thus $M[j]\gets \max\{M[j],r\}$.
The sketch is clearly idempotent and can be easily merged by taking the elementwise maximum of the registers.

The size of the HLL sketch is $O(m\log\log n)$ bits, as each register 
records an index of the $w$-bit hash value, which can be stored in $O(\log w)=O(\log\log n)$ bits.
The update time is $O(1)$, assuming the hash value can be computed in constant time.
In addition to the sketch, $w=O(\log n)$ auxiliary bits are needed to store the hash value~$x$,
yielding a total size of $O(m\log\log n+\log n)$ bits for the sketch.
We skip the estimation routine, as it is not relevant for this paper.

\subsection{Huffman coding}
The Huffman code~\cite{Huffman:1952} over $n$ symbols with probabilities $p_r$ is constructed  
by assigning each symbol~$r$ as a singleton tree with weight~$p_r$, and then recursively merging the two
trees with the least weights, by joining them with a node that records the sum of the total weights of the trees, 
until only tree remains.
If $p_1 \leq p_2 \leq \cdots \leq p_n$, that is, the probabilities are sorted, this can be done in $O(n)$ 
time~\cite{Leeuwen:1976}. 
The resulting tree is a full binary tree with $n-1$ internal nodes and $n$ leaves; the codeword is obtained by traversing 
the tree from the root to a leaf corresponding to the symbol, assigning \texttt{0} and \texttt{1} to the left and right
edges, and concatenating the bits. Let $L$ be a random variable for the length of the codeword.
It is well known that
the expected codeword length is at most the entropy of the distribution plus one bit, that is,
$\E[L]\leq H(R)+1=1-\sum_{r=1}^n p_r \log_2 p_r$.

\section{Algorithm and data structure}
\label{sec:algorithm}

In this section, we describe the algorithm and the associated data structure. At some points, we will make
forward-references to propositions that are proved in Section~\ref{sec:analysis}, 
but the description of the algorithm and data structure is self-contained, 
and does not require any knowledge of the analysis. The description here is aimed at the asymptotic analysis,
and a practical implementation would likely deviate on some points for improved leading constants, 
which will be addressed in Section~\ref{sec:practicalimplications}. 

\subsection{Huffman-Bucket Sketch data structure}
\begin{samepage}
The Huffman-Bucket Sketch data structure consists of
\begin{itemize}
    \item An array of $\frac{m}{B}$ buckets, each of which holds $B$ registers, 
    \item A global Huffman tree or codebook mapping between ranks and variable-length codewords, 
    \item A global cardinality estimate $\hat{n}$, and
    \item The cardinality estimate $\hat{n}_{\textrm{old}}$ at the time of the last Huffman tree construction.
\end{itemize}
\end{samepage}
Each bucket, in turn, consists of
\begin{itemize}
    \item An array of $B$ registers, encoded as variable-width Huffman codewords,
    \item An array containing a unary encoding of register codeword lengths,
    \item The bucketwise minimum rank $r_{\min}$ of the registers in the bucket, 
    \item The count of minimum-rank registers $c_{\min}$, and
    \item The bucketwise cardinality estimate $\hat{n}_b$.
\end{itemize}
The data structure is simply initialized with estimates $\hat{n}_b=\hat{n}=0$ for all buckets $b$, with a trivial
Huffman tree,\footnote{When $n=0$, the Huffman tree is not well defined because $\Pr[R_j>0]=0$, but
applying Huffman's algorithm anyway yields a simple comb, corresponding to an unary encoding of $r$.}
 encoding all register values to the codeword corresponding to rank zero, and setting $r_{\min}=0$ and 
$c_{\min}=B$ for all buckets.

Let $B=O(\log n)$.
Suppose that the codeword lengths are $\ell_1,\ell_2,\ldots,\ell_B$, and let $L=\sum_{j=1}^B \ell_j$ be the total 
length of the codewords in the bucket.
We will later show in Proposition~\ref{prp:bucketlength} that the codeword array has at most $L=O(\log n)$ bits, 
with high probability.
The unary encoding of $\ell_j$ is \texttt{1}$^{\ell_j}$\texttt{0}, 
so therefore the unary array takes $\sum_{j=1}^B (\ell_j+1)=O(\log n)$ bits.
The minimum rank and the count of minimum rank registers require $O(\log \log n)$ bits, 
and the cardinality estimate requires $O(\log n)$ bits.
Thus, the total size of a bucket is $O(\log n)$ bits.
This, in turn, means that the total size of the bucket array is $\frac{m}{B}\cdot O(\log n)=O(m)$ bits.

By Proposition~\ref{prp:bucketdistribution}, 
the registers of the buckets have the same distribution as the overall sketch, so we only need to maintain
one global Huffman tree for the entire sketch, which is constructed using the global cardinality estimate $\hat{n}$.
Furthermore, by Proposition~\ref{prp:combinedestimates}, the global cardinality estimate is just a better estimate 
of the latent cardinality parameter that unambiguously determines the distribution of the register values in the buckets
for a fixed~$m$, and thus the
Huffman tree. The estimate~$\hat{n}$ represents our best knowledge of~$n$. 
In the case of very small cardinalities for a bucket,
the minimum rank and the count of minimum rank registers allow us to simply fall back to linear counting
as a small-cardinality correction, like in~\cite{FlajoletFGM:2007}.

The Huffman tree is a full binary tree of $O(\log n)$ nodes, and as such can be encoded in $O(\log n)$ 
bits that store the structure of the tree, mapping each node to whether it is a leaf or an internal node.
The size of the global cardinality estimates is $O(\log n)$ bits. We can thus state the following theorem.
\begin{theorem} 
    The total size of the Huffman-Bucket Sketch data structure is $O(m+\log n)$ bits, which is 
    optimal~\cite{AlonMS:1999,IndykW:2003,KaneNW:2010,Woodruff:2004}.
\end{theorem}
\begin{proof}
    This follows directly from the definition of the data structure and
    Proposition~\ref{prp:bucketlength}.
\end{proof}

We can improve the running times of operations on the sketch if we assume $m=\Omega(\log^2 n)$.
This allows us to replace the Huffman tree with two lookup tables: an array that maps ranks to codewords,
taking $O(\log^2 n)$ bits, and a hash table that maps codewords to ranks, taking as well $O(\log^2 n)$ bits.
The total size of the sketch becomes then $O(m)$ bits.

\subsection{Operations on the Huffman-Bucket Sketch}

The sketch supports four operations: peek, poke, insert, and merge. 
The peek operation takes a bucketwise index pair $(b,j)$ and returns the value of the $j^\textrm{th}$ register in the 
bucket~$b$. The poke operation takes a rank value~$r$, a bucketwise index pair $(b,j)$, 
and sets the $j^\textrm{th}$ register in the bucket~$b$ to $r$. 
The insert operation takes an element~$y\in\mathcal U$ and updates the sketch with $y$.
The merge operation takes two sketches and merges them into one sketch.
The operations are described in pseudocode in Appendix~\ref{app:pseudocode}. 
\subsubsection{Peek (lookup register value)}
Registers are indexed by a bucketwise index pair $(b,j)\in[m/B]\times[B]\simeq [m]$.
To obtain the value of the register, without making any assumptions on $n$ and $m$,
we need to locate the codeword in the array, and decode it using the Huffman tree.
Locating the start and the length of the codeword using the unary array takes 
$O(B)=O(\log n)$ time, and decoding the codeword takes $O(\log n)$ time,
as we may need to traverse the entire Huffman tree in the worst case.
Thus the peek operation can take $O(\log n)$ time in the worst case.

If we assume $m=\Omega(\log^2 n)$, the decoding becomes a constant-time lookup.
Furthermore, if we assume that the word length of the machine is $\Omega(\log n)$, 
we can locate the codeword and its length in constant time, therefore making the
peek operation run in $O(1)$ time.

\subsubsection{Poke (set register value)}
Given a new rank value~$r$ to be set in the $j^\textrm{th}$ register of the bucket~$b$, 
we need to encode $r$ into a codeword and replace the old codeword with the new codeword.
This may require us to shift the right-hand side of the bucket to accommodate the new codeword if it
differs in length.
Without any assumptions on $n$ and $m$, 
the encoding and shifting take $O(\log n)$ time.
However, if we assume $m=\Omega(\log^2 n)$ and 
that we can modify $\Omega(\log n)$ bits in constant time, 
then encoding, and shifting
operations become $O(1)$.

\subsubsection{Insert (update sketch with a new element)}
Given an element~$y\in\mathcal U$, use three distinct hash functions $h_b : \mathcal U \to [m/B]$, 
$h_j : \mathcal U \to [B]$, and $h_r : \mathcal U \to [2^w]$ (or one hash function and select an appropriate number of
independent bits) to compute the bucket index~$b=h_b(y)\in[m/B]$, 
the register index~$j=h_j(y)\in[B]$, and the rank~$r=\rho(h_r(y))\in [w]$.

If $r \leq r_{\min}$ for the bucket~$b$, the update becomes a nop.
Otherwise, we need to peek the old register value~$r_{\mathrm{old}}$.
If $r_{\mathrm{old}} \geq r$, this again becomes a nop. An actual update is only
triggered if $r_{\mathrm{old}} < r$. The update is performed by poking the register $(b,j)$
with the new rank~$r$. After the poke, if $r_{\mathrm{old}} = r_{\min}$, we need to set $c_{\min}\gets c_{\min}-1$, 
and if $c_{\min}$ becomes zero, we need to recompute $r_{\min}$ and $c_{\min}$ by peeking all registers in the bucket.
We also need to update the local cardinality estimate $\hat{n}_b$ and the global cardinality estimate $\hat{n}$.
If the global estimate changes too much from $\hat{n}_{\textrm{old}}$, 
we need to reconstruct the Huffman tree and re-encode all buckets.

Without any assumptions about $n$ and $m$, 
the peek and poke operations take $O(\log n)$ time,
updating $r_{\min}$ takes $O(\log^2 n)$ and
updating the estimate(s) 
takes $O(\log n)$ time.
By unimodality of the distribution (see Section~\ref{sec:unimodality}),
it is easy to determine the unique mode in $O(1)$ time, 
sort the probabilities in $O(\log n)$ time, and construct the Huffman tree in $O(\log n)$ time.
Re-encoding all buckets takes $O(m\log n)$ time.
Worst-case running time is thus $O(m\log n+\log^2 n)$.
We will later show in Theorem~\ref{thm:amortizedupdates} that this leads to an amortized 
running time of $O(1)$ per insertion over the entire stream,
assuming $m=O(n/\log^3 n)$. The intuitive reason for the bound on $m$ is that if $m$ is too large, 
the sketch becomes too sparse, 
as almost all registers are zero; in this regime, it would make more sense to use linear counting anyway~\cite{WhangVT:1990}.

If we assume $m=\Omega(\log^2 n)$ and that we can modify $\Omega(\log n)$ bits in constant time, 
then peek, poke, and estimate-update\footnote{At least if we use the vanilla HLL estimator, which is just a harmonic mean times a constant.} 
become $O(1)$, and updating $r_{\min}$ becomes $O(\log n)$.
Reconstructing the Huffman tree remains $O(\log n)$,
but re-encoding the codewords becomes only $O(m)$.
The best case is thus $O(1)$ and worst case is $O(m)$, with the Huffman tree construction subsumed the re-encoding of all buckets.
In this regime, as shown in Theorem~\ref{thm:amortizedupdateswithconstantwords} we also have an amortized running time of $O(1)$ 
per insertion over the entire stream, but we can easily lighten the assumption on $m$ to assuming $m=O(n/\log^2 n)$.
So, the algorithm works best in the regime where $m$ is sufficiently larger than the hash values,
but substantially smaller than $n$ to avoid sparsity, which corresponds to the real-world regime.

\subsubsection{Merge (combine two sketches into one sketch)}
Mergeability in itself is inherited directly from HLL,
as the sketches are lossless compressions of HLL sketches.
For a merge operation, we need to necessarily decode all register values,
but we can potentially reuse the Huffman tree, and we may not need to re-encode all values
if the new estimate is not too large. The best case scenario is when either one of the sketches
has a considerably larger global estimate than the other, in which case the smaller sketch hardly contributes
anything, we reuse the larger sketch's Huffman tree, and mostly just copy codewords as-is.
The worst case scenario is when the two sketches have similar estimates but represent disjoint sets of elements,
in which case the new estimate is approximately twice as big as the maximum of the old ones, necessitating
Huffman tree reconstruction, and re-encoding of all registers.

The merge operation proceeds in two passes. First, we need to decode all register values to compute 
the new bucketwise and global estimates, taking $O(m\log n)$ time without assumptions on $n$ and $m$.
We then construct the new Huffman tree if necessary, taking $O(\log n)$ time.
Finally, we need to re-encode all registers with the new Huffman tree, which takes $O(m \log n)$ time.

If we can assume $m=\Omega(\log^2 n)$ and that we can access $\Omega(\log n)$ bits in constant time, 
the decoding becomes a constant-time lookup, and the re-estimation and re-encoding become $O(m)$ time.
Therefore, under these assumptions, the merge operation takes $O(m)$ time, as the Huffman tree construction
is subsumed by the $O(m)$ time of re-encoding all registers.

\begin{theorem}
\label{thm:mergetime}
    The merge operation takes $O(m\log n)$ time in the worst case. 
    However, assuming $m=\Omega(\log^2 n)$ and that $\Omega(\log n)$ bits can be accessed in constant time,
    the merge operation can be implemented in $O(m)$ time.
\end{theorem}
\begin{proof}
    The result follows directly from the description of the merge operation above.
\end{proof}

\section{Analysis in the poissonized balls-and-bins model}
\label{sec:analysis}

We now analyze the algorithm in a standard balls-and-bins model. 
We will be assuming a Poisson approximation of the distribution throughout this section, 
treating the number of balls~$n$ as a random variable and all of the~$m$ registers as independent and identically distributed,
receiving $\lambda = \frac{n}{m}$ balls on average. We shall refer to $\lambda$ as the \emph{load factor}.

\subsection{Rank distribution}
\label{sec:distribution}

Assuming~$m$ registers, we shall consider the rank distribution of the fixed register~$j\in[m]$ after~$n$ insertions.
Consider a matrix of bins with~$m$ rows and infinitely many columns. We independently throw balls into the bins, such
that each ball is thrown into a random row~$j$ uniformly at random, and to column~$r\geq 1$ with probability $2^{-r}$.
In the Poisson model, this means that the number of balls~$N_{j,r}$ in the bin $(j,r)$
is $N_{j,r}\sim \mathrm{Pois}(\lambda 2^{-r})$, independently for all $j\in[m]$ and $r\geq 1$.
If we consider the matrix of indicator variables $I_{j,r}=\llbracket N_{j,r}>0 \rrbracket$,
this corresponds to the FM85 matrix~\cite{FlajoletM:1985} that underlies the entire LogLog family of algorithms.

The register value $R_j=\max_{r\geq 1} \{ r \mid N_{j,r} > 0 \}$ is thus the maximum~$r$ such that $N_{j,r}>0$, or zero if there is no such~$r$. By the properties of the Poisson distribution, we have the following distribution for the register value:
\begin{itemize}
    \item $\Pr[R_j\leq r] = \prod_{k=r+1}^\infty \Pr[N_{j,k}=0] = e^{-\lambda 2^{-r}}$ for $r\geq 1$,
    \item $\Pr[R_j=0] = \Pr[R_j\leq 0] = e^{-\lambda}$,
    \item $\Pr[R_j=r] = \Pr[R_j\leq r] - \Pr[R_j\leq r-1] = e^{-\lambda 2^{-r}} - e^{-\lambda 2^{-(r-1)}} = e^{-\lambda 2^{-r}} - e^{-2\lambda 2^{-r}} = e^{-\lambda 2^{-r}}(1 - e^{-\lambda 2^{-r}}) $ for $r\geq 1$,
    \item $\Pr[R_j>r] = 1 - \Pr[R_j\leq r] = 1 - e^{-\lambda 2^{-r}}$ for $r\geq 0$, and
    \item $\Pr[R_j\geq r] = \Pr[R_j>r-1] = 1 - e^{-\lambda 2^{-(r-1)}}=1-e^{-2\lambda 2^{-r}}$ for $r\geq 1$.
\end{itemize}
As we make heavy use of the definition for $\Pr[R_j= r]$, we denote this probability simply by $p_r$ when $\lambda$ is clear from context and $j$ is irrelevant.
It is also sometimes convenient to express the probabilities as exponential functions, so we denote
$x_r = x_r(\lambda) = \lambda2^{-r}$, and $g(x) = e^{-x}(1 - e^{-x})$, so that $p_r = g(x_r) = e^{-x_r}(1 - e^{-x_r})$.

Suppose we divide the registers into $m/B$ buckets of $B$ registers each.
Then, the register values still have the same distribution.
\begin{proposition}
    \label{prp:bucketdistribution}
    Consider the bucket of $B$ registers that receives $n/(m/B)$ balls in expectation.
    Then, the distribution of the register values in the bucket is the same as the full sketch 
    with $n$ balls and $m$ registers.
\end{proposition}
\begin{proof}
    Consider a sketch with $B$ registers and $n/(m/B)$ balls in expectation.
    By direct computation we get $\lambda = \frac{n/(m/B)}{B} = \frac{n}{m}$, 
    which matches the full sketch.
\end{proof}

\begin{figure}[t]
    \begin{center}
        \includegraphics[width=\textwidth]{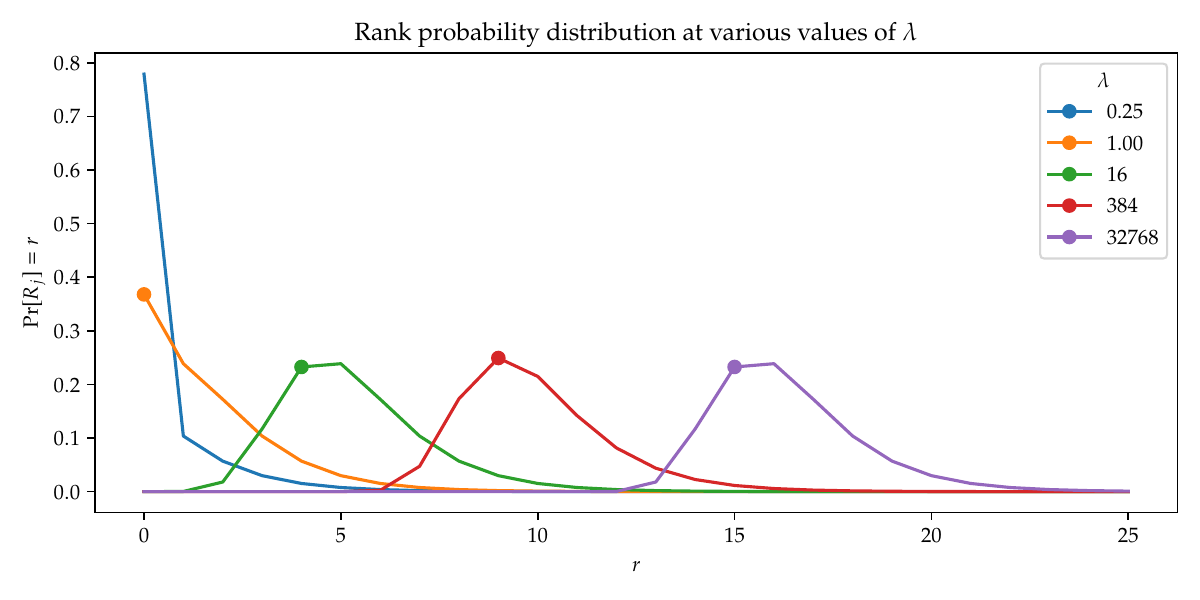}
        \caption{Register value distribution for various $\lambda$,
        with the marker representing $r^*=\lceil\log_2\lambda\rceil$.}
        \label{fig:rankdistribution}
    \end{center}
\end{figure}

\subsection{Unimodality and tails}
\label{sec:unimodality}
The distribution is very strongly concentrated around $r^*=\lceil\log_2\lambda\rceil$.
Intiuitively, this makes sense: each register essentially estimates the logarithm of the number of elements
in their substream, which is $\log_2 \lambda$ in expectation.
The rank probability has sharply decaying tails, but is not symmetric.
The right tail decays exponentially, whereas the left tail decays doubly exponentially, which is shown in Appendix~\ref{app:tailbounds}.
In fact, this is the reason why the harmonic mean of HLL improves estimation accuracy over the arithmetic
mean of LogLog: the harmonic mean penalizes the outliers in the right tail.
The rank distribution is illustrated in Figure~\ref{fig:rankdistribution} for various values of $\lambda$,
together with the marker representing $r^*$.
The shape indicates that the distribution is unimodal.
\begin{restatable}{proposition}{prpmode}
    \label{prp:mode}
    The rank distribution is unimodal and has a mode at either $r^*-1$, $r^*$, or $r^*+1$, 
    where $r^*=\lceil\log_2\lambda\rceil$.
\end{restatable}
\begin{proof}
    See Appendix~\ref{app:modeproof}.
\end{proof}

\subsection{Combined estimates}
It was already established in Proposition~\ref{prp:bucketdistribution} that the buckets have the same distribution as the full sketch.
This implies that the global cardinality estimate $\hat{n}=\sum_{b=1}^B \hat{n}_b$ is simply a more accurate estimate of the latent cardinality $n$.
\begin{restatable}{proposition}{prpcombinedestimates}
    \label{prp:combinedestimates}
    Let $\hat{n}=\sum_{b=1}^{m/B} \hat{n}_b$ be the sum of the bucketwise estimates, 
    each with a relative standard error of $O(1/\sqrt{B})$.
    Then, $\hat{n}$ has relative standard error of $O(1/\sqrt{m})$.
\end{restatable}
\begin{proof}
    See Appendix~\ref{app:combinedestimatesproof}.
\end{proof}

\subsection{Bucket size}
For $j\in[B]$, let $L_j$ be the length of the $j^\textrm{th}$ codeword in the bucket.
We shall now prove that the total size of the codeword array $L=\sum_{j=1}^B L_j$ in each bucket is $O(B)=O(\log n)$ bits with high probability.
The idea is to make the observation that (i) the entropy of the register value distribution is $O(1)$ bits, 
(ii) by optimality of the Huffman code, this implies that the expected length of a codeword is $O(1)$ bits, and then (iii)
use a known upper bound on the worst-case length of a Huffman codeword together with the Chernoff bound
to show that $L=O(\log n)$ bits with high probability.
\begin{samepage}
\begin{restatable}{proposition}{prpbucketlength}
    \label{prp:bucketlength}
    Let $B=\alpha \log n$ for some constant $\alpha>0$. 
    Then there exist constants $\Gamma,\gamma>0$, independent of $n$ and $m$,
    such that for all $t\geq 0$, $\Pr[L \geq \Gamma\log n + t] \leq e^{-\gamma t}$.
\end{restatable}
\begin{proof}
    See Appendix~\ref{app:bucketlengthproof}.
\end{proof}
\end{samepage}

\subsection{Number of Huffman tree reconstructions}
\label{sec:treechanges}
We will now show that the number of Huffman tree reconstructions over $n$ insertions is $O(\log n)$.
The intuition behind the result is that is unimodal per Proposition~\ref{prp:mode}, and 
the tails are very steep per Appendix~\ref{app:tailbounds}, so the distribution is very rigid in the tails and 
there can be no changes to the Huffman tree in the parts corresponding to the tails, as the tree collapses into
a chain or a ``comb'' structure with a predefined merge order. Moreover, since the left tail decays much faster
than the right tail, we can partition the distribution into three parts: the left tail, 
a constant-width \emph{band} around the mode, and the right tail. The left tail must collapse into a comb
before the right tail has any effect on the tree; and the right tail must collapse into a comb before the
band has any effect on the tree, so all nontrivial changes to the tree must be in the band.
As the band is determined by the mode, and the mode approximates the base-2 logarithm of~$n$,
this indicates that there are nontrivial changes roughly when $\lambda$ crosses a power of 
two, shifting the distribution to the right. This can only change $O(\log n)$ times, 
so the number of changes to the Huffman tree is $O(\log n)$.
\begin{restatable}{proposition}{prphuffmantreechanges}
    \label{prp:huffmantreechanges}
    Let $\lambda_n=\frac{n}{m}$ for some fixed $m$, and $N$ an upper bound on $n$. 
    The total number of tree changes as $\lambda$ varies over $\lambda_1, \lambda_2, \ldots, \lambda_N$ is $O(\log N)$.
\end{restatable}
\begin{proof}
    See Appendix~\ref{app:huffmantreechangesproof}.
\end{proof}

\subsection{Amortized updates}
We will now analyze the amortized update costs. There are three kinds of updates:
(i) ordinary updates,
(ii) minimum-rank updates, and
(iii) tree updates.
For a total of $N$ updates, the number of tree updates was bounded to be at most $O(\log N)$ in Section~\ref{sec:treechanges},
so it remains to bound the number of ordinary updates and minimum-rank updates.
The intuition is that there are not too many such updates if $\lambda$ is sufficiently large. This is because
once the sketch becomes ``saturated'' enough, almost no new element can trigger an update, as the rank distribution 
is geometrically decreasing. Conversely, 
if $m$ is close to $n$, the sketch is very sparse and most registers are zero, meaning we are better off using linear counting~\cite{WhangVT:1990} anyway.
\begin{restatable}{theorem}{thmamortizedupdates}
    \label{thm:amortizedupdates}
    Assuming an ordinary update takes $O(\log n)$ time, a minimum-rank update takes $O(\log^2 n)$ time, and
    a Huffman tree update followed by re-encoding of all codewords takes $O(m\log n)$ time,
    and furthermore assuming that $m=O(n/\log^3 n)$, the amortized cost of updates is O(1) over $n$ updates.
\end{restatable}
\begin{proof}
    See Appendix~\ref{app:amortizedupdatesproofs}.
\end{proof}
\begin{restatable}{theorem}{thmamortizedupdateswithconstantwords}
    \label{thm:amortizedupdateswithconstantwords}
    Assuming that we can access $\Omega(\log n)$ bits in constant time and $m=\Omega(\log^2 n)$, so we can maintain a Huffman codebook,
    so we can assume an ordinary update takes $O(1)$ time, a minimum-rank update takes $O(\log n)$ time, and
    a Huffman tree update followed by re-encoding of all codewords takes $O(m)$ time,
    the amortized cost of updates is O(1) over $n$ updates if we further assume $m=O(n/\log^2 n)$.
\end{restatable}
\begin{proof}
    See Appendix~\ref{app:amortizedupdatesproofs}.
\end{proof}

\section{Generalizations}
\label{sec:generalizations}

So far we have focused the analysis of the HBS sketch specifically on the lossless compression of HLL registers.
However, the framework admits various generalizations, which we briefly outline here, without analyzing them in detail.

Perhaps the most immediate generalization is to change the basis of the rank or $\rho$ function.
Instead of looking at the location of the first 1-bit in the binary representation of the hash value, 
we could look at the location of the first nonzero digit in an arbitrary base-$b$ representation of the hash value.
This is not unlike the directions described in~\cite{Ertl:2024,PettieW:2021}.
Combining such an alternative rank function works directly with the HBS framework,
as the compression is agnostic to the actual symbols being compressed, and it should be possible to obtain
similar guarantees on the bucket behavior. The rank function could even be something completely different, as long
as it has a similar tail behavior.

Secondly, the current analysis works on the Huffman code, but there is no theoretical reason to stick to it, 
apart from its simplicity, ease of implementation, and ease of analysis.
We could consider other schemes such as arithmetic coding, provided they can be updated as efficiently,
and it can be shown that the codewords need not be updated too frequently.

Thirdly, we need not restrict ourselves to the HLL sketch and could consider
applying the same framework to other sketches, using the extra information
contained in the FM85 matrix~\cite{PettieW:2021}, like, for example, in the ULL~\cite{Ertl:2024} or ExaLogLog~\cite{Ertl:2025}.

Lastly, it should be possible to apply the framework to other problems as well, such as compressing the 
Count-Min sketch~\cite{CormodeM:2005}, which estimates the frequencies of individual elements in the data stream,
maintaining an array of integer counters using hash functions. 
The key complication for the Count-Min sketch would be that there would have to be some kind of model for the
distribution of counter values, as we cannot abstract the data-dependence away like in the case of HLL.
However, it might be possible to treat, for example, the rows of the CMS as buckets to be compressed with mild 
assumptions on the data, or alternatively learn the distribution together with the data. However, we leave this 
as an open question for now.

\section{Practical implications}
\label{sec:practicalimplications}

\begin{figure}[t]
    \begin{center}
        \begin{tabular}{@{}c@{}c@{}}
            \includegraphics[width=0.45\textwidth]{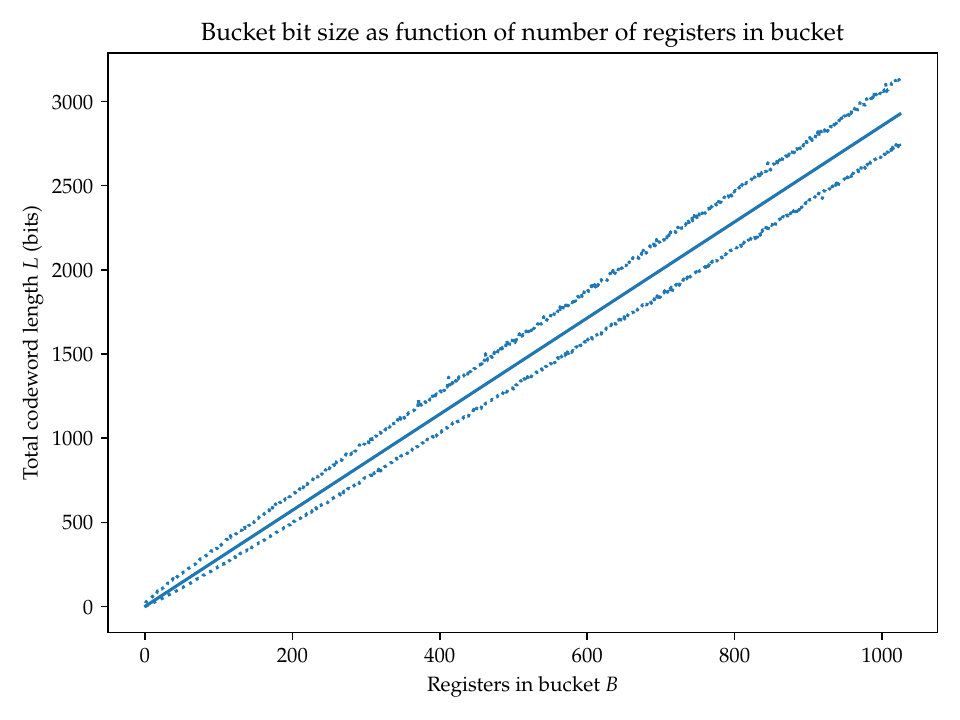} & 
            \includegraphics[width=0.45\textwidth]{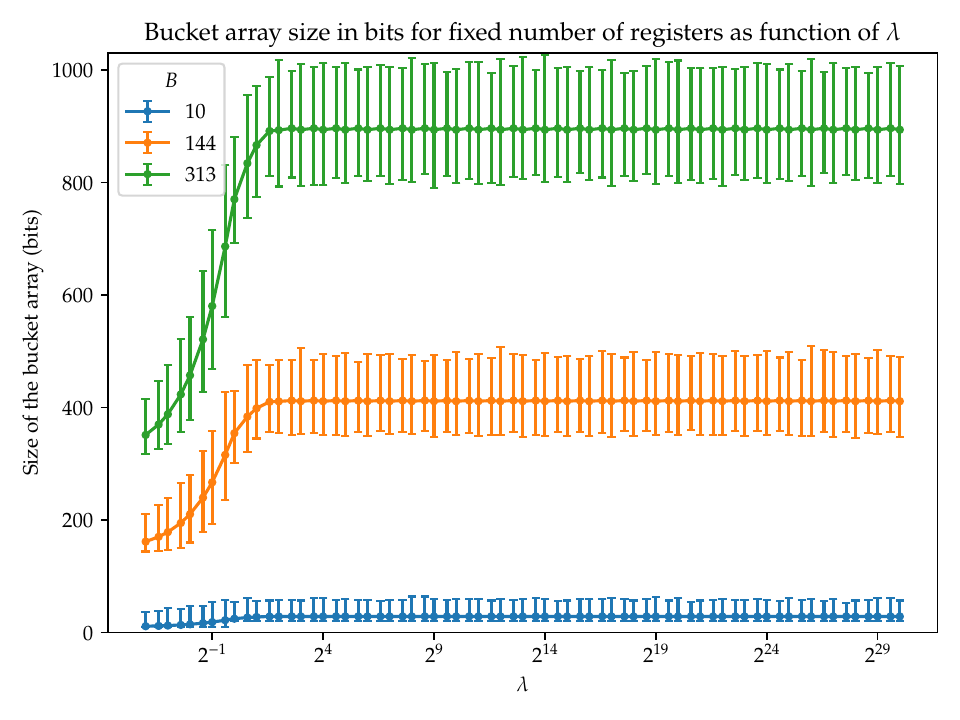}\\
            (a) & (b)
        \end{tabular}
        \caption{(a) The bit size of the bucket array or total codeword size~$L$ as function of the number of registers
        in the bucket~$B$. The solid line represents the mean of one million repetitions, and the dashed lines
        the minimum and maximum sizes.
        (b) The size of a bucket with a fixed number of registers~$B$ as a function of the load factor~$\lambda$
        over one million random sketches. 
        The marker represents the mean, the error bars represent the minimum and maximum sizes.}
        \label{fig:bucket_size}
    \end{center}
\end{figure}

\subsection{Implementation concerns}
\label{sec:implementationconcerns}
The HBS sketch was designed with practical implementation in mind. The idea of bucketing arose from the idea
that the buckets should be small enough to be compressed either into single machine words or a single cache line.
The analysis of Section~\ref{sec:analysis} is of limited value for a practical implementation, as the world is not asymptotic and constants matter.
This leads to, for example, that we cannot really use the idea of a unary index into the array. However, 
$\log n$
\emph{is} a constant for all intents and purposes for all realistic values of $n$, 
as, in practice, we use 64-bit hash functions. 
This makes 
accesses constant time in that
they do not really depend on the number of registers or the cardinality. If we wanted to improve the access time,
we could provide a within-bucket index of the bit locations of $\lceil \log_2 B \rceil-1$ equally-spaced
registers, reducing the need from decoding $\Theta (B)$ registers to $\Theta (B/\log B)$ registers, 
which is a significant improvement in practice, as $B$ is typically small. At the same time, it is not hard to see
that the size of the sketch is dominated by the codewords of the buckets, so, for example, a codebook for the 
Huffman code is minuscule in comparison.

Figure~\ref{fig:bucket_size}a shows the size of the total codeword length $L$ of a bucket as a function the number
of registers in the bucket~$B$. The values were computed by drawing one million random sketches
for each~$B$, with $n=2^{30}$, $m=2^{15}$, meaning $\lambda = 2^{15}$,
and encoding the buckets with the Huffman code. The solid line shows the mean and the dotted lines 
the minimum and maximum sizes. As the analysis of Section~\ref{sec:analysis} suggests, 
the size grows linearly in $B$, as do the minimum and maximum, meaning the size of the codeword array is determined by 
the number of registers in a bucket. The analysis also indicates that, due to the shape of the distribution as in 
Figure~\ref{fig:rankdistribution}, as the code is determined by a constant-width band around the peak around 
$r^* = \lceil \log_2 \lambda \rceil$, the bit size of the bucket is 
independent of $\lambda$, assuming $\lambda$ is sufficiently large that the left tail can emerge. If the
left tail does not emerge, then the size of the bucket is only smaller. Thus, as $\lambda\to\infty$, 
the size of the bucket is asymptotically constant, depending only on $B$.
Figure~\ref{fig:bucket_size}b shows this behavior with various values of $\lambda$ for three different values of $B$. 
The marker represents the mean, and the error bars represent the minimum and maximum sizes.

Thus, what one would do in practice, is choose a bit budget~$b_B$ for the codeword array, and then choose
$B$ such that there is sufficient slack that it is improbable for any bucket to overflow. 
We provide a rudimentary analysis here for 
the arbitrary case of $n=2^{30}$ and $m=2^{15}$ with values drawn from the random sketches of Figure~\ref{fig:bucket_size}a, 
without attempting to rigorously optimize the parameters, but just to get a sense of the practical implications.
We consider
three 
such budgets:
64 bits, 512 bits, and 1024 bits, corresponding to one machine word, one cache line, and two cache lines, respectively.
If we choose suitable values of $B$ based on the data from Figure~\ref{fig:bucket_size}a, 
we get values 10, 144, and 313, respectively.

In addition to this array for each bucket,
we need to store the minimum rank $r_{\min}$ and the count of minimum-rank registers $c_{\min}$
for each bucket, which take 6 and $\lceil \log_2 B \rceil$ bits, respectively.
We also need to store the global estimates $\hat n$ and $\hat{n}_{\text{old}}$, and the Huffman tree.
Unlike the theoretical description, we do not want to store local estimates because the 64 bits required for that would
be too much, but this is fine if we use the vanilla estimator~\cite{FlajoletFGM:2007} for the global $\hat{n}$,
as it is easily updated in constant time because it is simply a harmonic mean.
The Huffman tree takes only 127 bits if we just store the indicator bits for the leaves and the internal nodes.
Since there are $\left\lceil \frac{m}{B} \right\rceil$ buckets, 
the total size of the sketch is $\left\lceil \frac{m}{B} \right\rceil \cdot (b_B + 6 + \lceil \log_2 B \rceil) + 127$ bits.
The total sizes are shown in Table~\ref{tab:memoryusage}, under the \emph{small sketch} column.
It is obvious that the total size decreases as the bit budget increases, representing a tradeoff between
the size of the sketch and the update times.

For more streamlined updates, we could store the index into a logarithmic number of registers within a bucket,
requiring $\left(\lceil \log_2 B \rceil-1\right)\cdot \lceil \log_2 n_B \rceil$ extra bits per bucket,
and store the Huffman code as a pair of lookup tables. The encoding table would require
$(63+6)\cdot 64 = 4416$ bits, as we need to store the codeword and its length for each possible value $r$,
and the decoding table would be a hash table of size at least
$(63+6+6)\cdot 64 = 4800$ bits, as we need to store the codeword and its length as a key, and also
the value $r$ as a value. The total sketch sizes are given in Table~\ref{tab:memoryusage} under the \emph{big sketch} 
column.

As a final observation, it is possible to further reduce the size of the sketch in practice by observing that we are
using \emph{a lot of} slack here, and there are only very few buckets that actually need the bit budget we are using.
Intuitively, the overflows occur when (usually an individual) surprisingly large rank value ends up in a bucket, 
requiring a very long codeword.
It might be possible to use an adaptive scheme where most buckets are stored with fewer bits, and we allocate a handful
of \emph{big buckets} with larger sizes that can accommodate the small number of overflowing buckets.

\begin{table}[t]
    \caption{Memory usage and Memory-Variance Product (MVP) of the sketch at $m=2^{15}$ for different bit budgets.}
    \label{tab:memoryusage}
    \begin{tabular}{rrrrrr}
        Bit budget & $B$ & Total small sketch size & Total big sketch size & MVP (small) & MVP (big) \\\hline
        64 bits    & 10  & 242679 bits             & 310800 bits           & 7.961       & 10.196 \\
        512 bits   & 144 & 119657 bits             & 143111 bits           & 3.926       & 4.695 \\
        1024 bits  & 313 & 108311 bits             & 125784 bits           & 3.553       & 4.127 \\
    \end{tabular}
\end{table}

\subsection{Memory-Variance Product}
The behavior of sketches such as HLL is often measured in terms of the Memory-Variance Product (MVP)~\cite{PettieW:2021}, 
which is the product of the memory usage and the relative variance of the estimator $\Var\left[\frac{\hat n}{n}\right]$.
The analysis of Section~\ref{sec:analysis} is asymptotic and ignores constants, so it does not directly give us
a meaningful result, apart from showing that the MVP is necessarily $O(1)$.
However, we can provide some rudimentary evidence that the algorithm is practical in these terms by modeling its
behavior numerically.

As we are not analyzing the estimators in this paper, we shall assume the relative variance is bounded by 
$1.075/m$~\cite{WangP:2023}.  As we observed in Section~\ref{sec:implementationconcerns}, 
the size of the bucket is determined solely by the choice of $B$, independent of $\lambda$, 
and the size of the sketch is dominated by the codeword array, so asymptotically,
the sketch size is $O(m)$ and we can compute the constants numerically by fixing the bit budget for the bucket.
This enables us to compute the MVP for different bit budgets, as shown in Table~\ref{tab:memoryusage}.
Comparing this to ExaLogLog~\cite{Ertl:2025}, 
even the simplistic scheme presented here without fully optimizing parameters and implementation details,
we can achieve an MVP that is comparable to the MVP of 3.67 of ExaLogLog without using
the extra information of the FM85 matrix and with the possibility of applying our framework on the framework
of~\cite{Ertl:2025}.

It should be noted that this is not a rigorous or complete analysis but a 
back-of-the-envelope calculation, giving confidence that the sketch is practical in terms of the MVP. 
A more rigorous analysis of the exact memory usage and the MVP is left for future work.

\section*{Acknowledgements}
I want to thank Rasmus Pagh and Martin Aum\"uller for helpful discussions and comments.
I also want to thank Tam\'as Sarl\'os for good ideas and the wonderful analogy of pulling ourselves up from the swamp by our hair, like Baron von M\"unchhausen.
Lastly, I want to thank my master students Labiba Karar Eshaba and Yasamin Fazelidehkordi, as some of the ideas
leading to this paper originated from our discussions on their master thesis project on a related but distinct
approach to cardinality estimation.

\bibliographystyle{alpha}
\bibliography{references}

\newpage

\appendix

\begin{center}
    \Huge Appendix
\end{center}

\section{Tail bounds for the distribution}
\label{app:tailbounds}
We shall prove the following bounds on the tails of the distribution, without claiming originality for the results.
\begin{lemma}
\label{lem:rankrighttail}
    Let $r^*=\lceil\log_2\lambda\rceil$. For any integer $k\geq 0$,
    we have $\Pr[R_j>r^*+k] \leq 2^{-k}$.
\end{lemma}
\begin{proof}
    Let $\theta = \lambda 2^{-r^*}\leq 1$.
    Thus, $\Pr[R_j>r^*+k] = 1 - e^{-\theta 2^{-k}} \leq \theta 2^{-k} \leq 2^{-k}$.
\end{proof}
\begin{lemma}
    \label{lem:ranklefttail}
    Assume $\lambda > 1$, and let $r^*=\lceil\log_2\lambda\rceil$. For any integer $k\geq 0$,
    we have $\Pr[R_j\leq r^*-k] \leq e^{-2^{k-1}}$.
\end{lemma}
\begin{proof}
    Let $\theta = \lambda 2^{-r^*}$. 
    We have $\theta 2^{k} \geq 2^{k-1}$ because $\lambda>1$ and $r^* \geq 1$.
    Thus, $\Pr[R_j\leq r^*-k] = 
    e^{-\theta 2^k} \leq e^{-2^{k-1}}$.
\end{proof}

\section{Proof of Proposition~\ref{prp:mode}}
\label{app:modeproof}

Let us start by considering the 
continuous analog of the distribution by replacing $r$ with a continuous variable $x$. 
Define $g : \mathbb{R}\to\mathbb{R}$ such that
$g(x) = e^{-x}(1-e^{-x})$.
\begin{lemma}
    \label{lem:unimodalityofg}
    The function $g$ is unimodal, with a unique maximum at $x=\ln 2$.
\end{lemma}
\begin{proof}
    We have $g'(x) = e^{-2x}(2-e^x)$. The only critical point is at $x=\ln 2$, which is a maximum,
    since $g'(x)>0$ for $x<\ln 2$ and $g'(x)<0$ for $x>\ln 2$.
\end{proof}

Let us consider the special case where the mode is at $r=0$ and the entire sequence is strictly decreasing.
\begin{lemma}
    \label{lem:modeat0}
     If $x_1=\frac{\lambda}{2}\leq \ln 2$, then the mode of the distribution is at $r=0$,
     and the entire sequence $p_r$ is strictly decreasing.
\end{lemma}
\begin{proof}
Consider the case $p_0\geq p_1 > p_2$, that is, the mode is at $r=0$, or jointly at $r=0$ and $r=1$.
This can happen only if $x_1=\frac{\lambda}{2}\leq \ln 2$, that is, $\lambda \leq 2\ln 2$.
As $\lambda\leq 2\ln 2$, we have $p_0 = e^{-\lambda} \geq e^{-2\ln 2} = \frac{1}{4}$, 
and $p_1 = e^{-\lambda/2}(1-e^{-\lambda/2}) \leq \frac{1}{4}$, so this condition guarantees that the mode is at $r=0$,
and the sequence is strictly decreasing after $p_1$.
\end{proof}

We can therefore now assume that $p_1>p_0$ and the mode is at $r\geq 1$.
Let us recall the statement of Proposition~\ref{prp:mode} and prove it.
\prpmode*
\begin{proof}
Let $r'$ be the least integer, such that $x_{r'} \leq \ln 2$, the maximum of $g$.
This is the ``crossing index'' where the distribution transitions from increasing to decreasing.
By the properties of $g$, we have that $p_r$ is strictly increasing when $r<r'$,
and strictly decreasing when $r > r'$. There is a potential special case if $x_r=\ln 2$,
as the sequence is then decreasing until $x_r$, but regardless of this, 
the sequence is nonincreasing for $r\geq r'$.
Therefore, the mode is either at $r'$ or at $r'-1$, whichever leads $x_r$ to be closer to $\ln 2$.

By the definition of $r^* = \lceil\log_2\lambda\rceil$, 
we have $x_{r^*} = \lambda 2^{-r^*} \leq 1$.
Furthermore, $x_{r^*-1} = \lambda 2^{-(r^*-1)}=2\frac{\lambda}{2^{r^*}}$, and
$x_{r^*+1} = \lambda 2^{-(r^*+1)}=\frac{1}{2}\frac{\lambda}{2^{r^*}}$.
Therefore, $\frac{1}{2} < x_{r^*} \leq 1$,
$1 < x_{r^*-1} \leq 2$, and
$\frac{1}{4}< x_{r^*+1} \leq \frac{1}{2}$.
Furthermore, $\frac{1}{2} < \ln 2 < 1$, 
so the crossing index must satisfy $r^* \leq r' \leq r^*+1$.
The mode must satisfy that it is either at $r'$ or $r'-1$.
Putting this together, we have that the mode is either at $r^*-1$, $r^*$, or $r^*+1$.
The result is thus proven with the special case of mode at zero handled by Lemma~\ref{lem:modeat0}.
\end{proof}

\section{Proof of Proposition~\ref{prp:combinedestimates}}
\label{app:combinedestimatesproof}
Let us recall the statement of Proposition~\ref{prp:combinedestimates} and prove it.
\prpcombinedestimates*
\begin{proof}
    We already established in Proposition~\ref{prp:bucketdistribution} that the distribution of the register values in a 
    bucket is the same as the distribution of the register values in the full sketch.
    This means that the global cardinality estimate~$\hat{n}=\sum_{b=1}^B \hat{n}_b$ 
    is simply a more accurate estimate of the latent cardinality~$n$.
    By the properties of the vanilla HLL estimator~\cite{FlajoletFGM:2007}, 
    each $\hat{n}_b$ has a relative standard error bounded by $O(1/\sqrt{B})$.
    Since we have an expected $n/(m/B)$ elements in the bucket,
    this implies that 
    \[
        \frac{\sqrt{\Var[\hat{n}_b]}}{\frac{n}{\frac{m}{B}}} =
        \frac{\sqrt{\Var[\hat{n}_b]}}{\frac{nB}{m}} =
        \frac{\sqrt{\Var[\hat{n}_b]}}{n}\cdot \frac{m}{B}
        \leq  \frac{K}{\sqrt{B}}
    \]
    for some constant~$K$ (1.04 according to~\cite{FlajoletFGM:2007}).
    Solving for $\Var[\hat{n}_b]$ gives $\Var[\hat{n}_b] \leq \frac{K^2 n^2 B}{m^2}$.
    Assuming independence of the bucketwise estimates, we have 
    $\Var[\hat{n}] = \sum_{b=1}^{m/B} \Var[\hat{n}_b] \leq \frac{K^2 n^2 B}{m^2}\cdot \frac{m}{B} = \frac{K^2 n^2}{m}$.
    Thus, the relative standard error of $\hat{n}$ is $\frac{\sqrt{\Var[\hat{n}]}}{n} \leq \frac{K}{\sqrt{m}}$.
\end{proof}

\section{Proof of Proposition~\ref{prp:bucketlength}}
\label{app:bucketlengthproof}

The proof proceeds in two parts. First, we show that the entropy of the register values is $O(1)$. This
implies that the expected codeword length of the Huffman code is also $O(1)$, from which we can bound the total size of the bucket array 
by $O(\log n)$, as claimed.

\subsection{Entropy}
We start by proving that the entropy of the register values is $O(1)$.
The proof is presented here for completeness, and we do not claim originality for the result,
as it is well known~\cite{Durand:2004,PettieW:2021}.
The following bound on $x_r$ is useful for bounding $\log p_r$.
\begin{lemma}
    \label{lem:xrbound}
    For $r\geq r^*=\lceil\log_2\lambda\rceil$, 
    $\Pr[R_j=r]=p_r$ satisfies $p_r\geq \frac{x_r}{2e}=\frac{\lambda 2^{-r}}{2e}$.
\end{lemma}
\begin{proof}
    For $r\geq r^*$, we have $e^{-x_r}\geq 1/e$. Also, $1-e^{-x_r} \geq \frac{x_r}{2}$.
    Putting these together gives $p_r\geq \frac{x_r}{2e}$.
\end{proof}
We can now prove the entropy bound.
\begin{lemma}
    \label{lem:entropy}
    The entropy of the register value $R_j$ is $H(R_j)=-\sum_{r=0}^{\infty} p_r \log_2 p_r = O(1)$.
\end{lemma}
\begin{proof}
    Consider
    \begin{equation}
        \label{eq:entropysplit}
        H(R_j) = \sum_{r=0}^\infty -p_r \log_2 p_r = (\log_2 e)\left( \sum_{r=0}^{r^*-1} -p_r \ln p_r + \sum_{r=r^*}^\infty -p_r \ln p_r
        \right) \, .
    \end{equation}

    We need to bound the two sums of Equation~\eqref{eq:entropysplit} separately.
    For the right-hand tail, since $1-e^{-x_r} \leq x_r$, we have $p_r \leq x_re^{-x_r}\leq x_r$.
    By Lemma~\ref{lem:xrbound}, $p_r \geq \frac{x_r}{2e}$.
    Thus, for $r\geq r^*$, we have
    \begin{equation}
        \label{eq:entropyprbounds}
        \frac{x_r}{2e} \leq p_r \leq x_r \, .
    \end{equation}

    Putting Equation~\eqref{eq:entropyprbounds} into the right-hand sum of Equation~\eqref{eq:entropysplit},
    setting $\theta=\lambda2^{-r^*} \leq 1$, so for $r=r^*+k$ we 
    have $x_r = \theta 2^{-k}$, 
    we get
    \begin{equation}
        \label{eq:entropyrighttail}
        \begin{split}
        -\sum_{r=r^*}^\infty p_r \ln p_r & \leq -\sum_{k=0}^\infty x_r \ln \frac{x_r}{2e} 
        = -\sum_{k=0}^\infty x_r \left(\ln x_r - \ln 2 - 1\right) \\
        & = -\sum_{k=0}^\infty x_r \ln x_r + \sum_{k=0}^\infty \left(\ln 2 + 1\right) \\
        & = -\sum_{k=0}^\infty \theta 2^{-k} \ln \left(\theta 2^{-k}\right) + \sum_{k=0}^\infty \theta 2^{-k}\left(\ln 2 + 1\right) \\
        & = -\theta \sum_{k=0}^\infty 2^{-k}\left(\ln \theta -k\ln 2\right) + \theta (\ln 2 + 1)\sum_{k=0}^\infty 2^{-k} \\
        & = -2\theta \ln \theta + 2\theta\ln 2 + 2\theta\ln 2 + 2 \theta = -2\theta \ln \theta + 4\theta\ln 2 + 2 \theta \, .
        \end{split}
    \end{equation}
    Since $-\theta\ln\theta\leq 1/e$ and $0\leq\theta\leq 1$, Equation~\eqref{eq:entropyrighttail} is 
    bounded by $\frac{2}{e}+4\ln 2+2<6$.

    For the left-hand tail and the left sum of Equation~\eqref{eq:entropysplit}, 
    we have $p_r\leq e^{-2^{k-1}}\leq 1/e$ for $r\leq r^*$,
    hence $-p_r\ln p_r \leq 2^{k-1} e^{-2^{k-1}}$. Thus,
    \[
        -\sum_{r=0}^{r^*-1} p_r \ln p_r \leq \sum_{k=1}^\infty 2^{k-1} e^{-2^{k-1}} < 1 \, .
    \]
    Putting this together, we have $H(R_j) < (\log_2 e)(6 + 1) = O(1)$.
\end{proof}
\begin{remark}
Numerically, $H(R_j)$ goes pretty fast to $\approx 2.83196$ for fixed $m$ as $n\to\infty$.
\end{remark}

\subsection{Bucket size}

We shall now prove that the size of the codeword array in each bucket is $O(B)=O(\log n)$ bits with high probability.
Let $L_j$ be the random variable for the length of the codeword for the $j^\textrm{th}$ register in the bucket, 
and let $L=\sum_{j=1}^B L_j$ be the total length of the codewords in the bucket. We need the following lemma
on the bound of the worst-case length of $L_j$. Correspondingly, let $L(r)$ be the length of the codeword encoding 
the rank value~$r$, that is, $L_j=L(R_j)$.
\begin{lemma}[\cite{AbuMostafaM:2000}, paraphrased]
    \label{lem:abumostafa}
    Suppose $L_j$ is the length of a Huffman codeword encoding a symbol with probability $p_r$.
    Then, $L_j \leq -1.44041\log_2 p_r$ in the worst case.
\end{lemma}
We need the following lemma.
\begin{lemma}
    \label{lem:prkbound}
    For $r\geq r^*$, we have $p_{r^*+k} \geq \frac{2^{-k}}{4e}$. 
\end{lemma}
\begin{proof}
    By the choice of $r^*$, we have $\frac{1}{2} < \frac{\lambda}{2^{r^*}}\leq 1$, so for $k\geq 0$,
    we have 
    $x_{r^*+k}=\frac{\lambda}{2^{r^{*}+k}} \geq \frac{1}{2} 2^{-k}=2^{-(k+1)}$. Combined with
    Lemma~\ref{lem:xrbound}, we have
    $p_r\geq \frac{x_{r^*+k}}{2e}=\frac{\lambda 2^{-{r^*+k}}}{2e} \geq 
    \frac{2^{-(k+1)}}{2e} = \frac{2^{-k}}{4e}$.
\end{proof}
Together, Lemmata~\ref{lem:abumostafa} and~\ref{lem:prkbound} imply for $k\geq 0$,
\begin{equation}
    \label{eq:codewordlengthbound}
    L(r^*+k) \leq -2\log_2 p_{r^*+k} \leq -2\log_2 \frac{2^{-k}}{4e} = 
    2k+2\log_2(4e) \, .
\end{equation}

It is well known (see, for example,~\cite[Chapter~5.7]{CoverT:2005}) that
Huffman codes are optimal in expectation, and $\E[L_j] \leq H(R_j) + 1$.
By Lemma~\ref{lem:entropy} and linearity of expectation, we thus have have $\E[L] = O(\log n)$.
However, we also need to show a stronger statement that $L$ is $O(\log n)$ with high probability.
We will do this by applying a standard Chernoff argument.
The Chernoff bound states the following:
Let $X_1,X_2,\ldots,X_n$ be independent and identically distributed random variables, 
with a shared moment-generating function $M(t)=\E[e^{t X_i}]$ that is finite for some $t>0$,
and let $X=\sum_{i=1}^n X_i$.
Then for any such $t>0$, we have
\begin{equation}
    \label{eq:chernoff}
    \Pr[X \geq a] \leq e^{-ta}\prod_{i=1}^n M(t) = e^{-ta} \left(\E[e^{t X_i}]\right)^n \, .
\end{equation}

To apply Equation~\eqref{eq:chernoff}, we need to establish that the
moment-generating function of $L_j$ is finite.
\begin{lemma}
    \label{lem:mgf}
    There exist constants $\tau>0$ and $M<\infty$,
    independent of $n$ and $m$, 
    such that $\E[e^{\tau L_j}] \leq M$.
\end{lemma}
\begin{proof}
    Fix $0<\tau<\frac{\ln 2}{2}$. We have
    \begin{equation}
        \label{eq:mgf1}
    \E[e^{\tau L_j}] = \sum_{r=0}^\infty e^{\tau L(r)} p_r 
    = \sum_{r=0}^{r^*-1} e^{\tau L(r)} p_r + \sum_{r=r^*}^\infty e^{\tau L(r)} p_r
    \, .
    \end{equation}
    The left-hand sum of Equation~\eqref{eq:mgf1} is finite, so it suffices to show that the right-hand sum is finite.
    By Equation~\eqref{eq:codewordlengthbound} and Lemma~\ref{lem:prkbound}, we have 
    \[
    \sum_{r=r^*}^\infty e^{\tau L(r)} p_r 
    = \sum_{k=0}^\infty e^{\tau L(r^*+k)} p_{r^*+k}
    \leq \sum_{k=0}^\infty e^{\tau (2k+2\log_2(4e))} \cdot \frac{2^{-k}}{4e}
    = \frac{e^{2\tau \log_2(4e)}}{4e} \sum_{k=0}^\infty \left(\frac{e^{2\tau}}{2}\right)^k < \infty \, ,
    \]
    since by our choice of $\tau$, we have $\frac{e^{2\tau}}{2} < 1$.
\end{proof}

Let us recall the statement of Proposition~\ref{prp:bucketlength} and prove it.
\prpbucketlength*
\begin{proof}
    Fix $\tau>0$ and $M$ such that $\E[e^{\tau L_j}] \leq M$, as guaranteed by Lemma~\ref{lem:mgf}.
    By the Chernoff bound of Equation~\eqref{eq:chernoff}, we have
    \[
    \Pr[L \geq s] \leq e^{-\tau s} \left(\E[e^{\tau L_j}]\right)^B
    \]
    Let $s=\Gamma\log n + t>0$ and $B=\alpha \log n$, so we have
    \[
    \Pr[L \geq \Gamma\log n + t] \leq 
    e^{-\tau (\Gamma\log n + t)} M^{\alpha \log n} \, .
    \]
    Choosing $\Gamma \geq \frac{\alpha\ln M+1}{\tau}$ yields by direct computation that
    \[
    \Pr[L \geq \Gamma\log n + t] \leq e^{-\tau t} \, ,
    \]
    proving the claim with $\gamma=\tau$.
\end{proof}

\section{Proof of Proposition~\ref{prp:huffmantreechanges}}
\label{app:huffmantreechangesproof}

This is a long and technical proof, so it has been broken into several subsections.
The main idea is that we establish that the tails of the Huffman tree are very rigid and become comb-like very fast, with a simple predefined merge order.
This implies that the only place where interesting tree changes can occur is within a narrow constant-width band around the mode.
We first analyze the right tail in isolation to establish the rigidity criterion. 
Then, we define the constant-width \emph{band} of large weights and the \emph{pool} of small weights.
The pool consists of the deep left and deep right tails.
We bound the mass of the left and right tails and show that the entire left tail is dominated by every element of the right tail;
similarly, the entire mass of the right tail is dominated by every element of the band.
Thus, the left tail has to collapse into a single node before it can interfere with the right tail,
and, likewise, the right tail has to collapse into a single node before it can interfere with the band.
As a consequence of Rolle's theorem, this implies that there can only be a constant number of tree changes, 
assuming $\lambda$ is constrained to a constant-width interval,
as the probability masses corresponding to subsets of nodes are finite-term sums of exponentials.
Once we have established the constant number of possible crossover points for a constant interval of $\lambda$,
the result follows by observing that there are only $O(\log n)$ such intervals.

\subsection{Basics}
Throughout this section, we assume that there exists an upper bound on the cardinality~$N$.
For fixed~$m$, it is obvious that the distribution defining the Huffman tree is unambiguously determined by~$n$.
Therefore, define $\lambda_n=\frac{n}{m}$, $x_r(\lambda)=\frac{\lambda}{2^r}$, and recall that
$p_0(\lambda)=e^{-\lambda}$ and $p_r(\lambda) = e^{-x_r(\lambda)}(1-e^{-x_r(\lambda)})$ for $r>0$.
Let $T(\lambda)$ denote the Huffman tree induced by the distribution $p_r(\lambda)$.
We shall analyze the set of \emph{change points} $\mathcal C_{N} = \{i \in [N-1] \mid T(\lambda_i)\neq T(\lambda_{i+1})\}$, 
that is, the set of indices where the Huffman tree changes. We aim to prove that $|\mathcal C_N|=O(\log N)$.

For $\lambda>0$, let $r^*(\lambda) = \lceil\log_2 \lambda\rceil$. Define the intervals
$I_0 = [0,1]$, and for $k\geq 1$, $I_k=(2^{k-1},2^k]$.
\begin{lemma}
    \label{lem:interval}
    If $k\geq 1$ and $\lambda\in I_k$, then $r^*(\lambda)=k$ and $2^{-k}\lambda\in (1/2,1]$.
\end{lemma}
\begin{proof}
    Immediate from $2^{k-1} < \lambda \leq 2^k$.
\end{proof}
We also define $g(x)=e^{-x}(1-e^{-x})$ for $x>0$. Then, $p_r(\lambda) = g(x_r(\lambda))$ for $r\geq 1$ and $p_0(\lambda)=e^{-\lambda}$.

\subsection{Right tail in isolation}
We will start by analyzing the right tail \emph{in isolation}, that is, we just ignore the left tail.
We show that the part of the Huffman tree corresponding to the right-hand
tail is very rigid and collapses into a chain or a ``comb'' structure with a simple predefined merge order.
The following lemma sets the criterion for the comb structure to occur.
\begin{lemma}
    \label{lem:combcriterion}
    Let $w_0\geq w_1 \geq \cdots \geq w_M >0$ be a sequence of \emph{weights}.
    If for all $i\leq M-2$, we have 
    \begin{equation}
        \label{eq:combcriterion}
        w_i \geq w_{i+1}+w_{i+2} \, ,
    \end{equation} 
    then the Huffman tree for 
    the distribution defined by $w_i$ has a comb structure with a predefined merge order.
\end{lemma}
\begin{proof}
    The Huffman algorithm merges the two least probable symbols (or inner nodes) at each step.
    By the criterion in \eqref{eq:combcriterion}, the newly formed sum is forced to be less than or equal to the 
    next weight, so the result follows by induction.
\end{proof}

\begin{definition}
For $x>0$, define the ratio function $h(x)$ as
\[
    h(x) = \frac{g(x/2)}{g(x)}
    = \frac{e^{-x/2}(1-e^{-x/2})}{e^{-x}(1-e^{-x})}
    = \frac{e^{x/2}}{1+e^{x/2}}\, .
\]
\end{definition}
\begin{lemma}
    \label{lem:hxstrictlyincreasing}
    For $x>0$, $h(x)$ is strictly increasing.
\end{lemma}
\begin{proof}
    We have $h'(x) = \frac{e^x(2+e^{x/2})}{2\left(1+e^{x/2}\right)^2} > 0$.
\end{proof}

Let $q_* = \frac{\sqrt{5}-1}{2}\approx 0.618034$.
Observe that 
\begin{equation}
    \label{eq:qpqsq}
    q_*+q_*^2=1 \, .
\end{equation}
Choose $\eta>0$ such that \[
\eta \leq 2\ln\left(\sqrt{5}-1+\sqrt{6\sqrt{5}-2}\right)-4\ln 2 \approx 0.285985 \, .
\]
Importantly, this guarantees that $h(\eta) \leq q_*$.
Fix integer $K'_R\geq 1$ such that $2^{-K'_R}\leq\eta$.
We can now show that the right-hand tail of the distribution is dominated by the last element before the tail,
beyond the cut-off point $r^*(\lambda)+K'_R$.
\begin{lemma}
    \label{lem:righttaildominance}
    Let $k\geq 1$, $\lambda\in I_k$, and $r\geq k+K'_R$. 
    Then, $p_r(\lambda) \geq p_{r+1}(\lambda) + p_{r+2}(\lambda)$.
\end{lemma}
\begin{proof}
    By Lemma~\ref{lem:interval}, we $\lambda 2^{-k} \leq 1$, 
    so for $r\geq k+K'_R$, we have by the choice of $K'_R$ that
    \[ 
    x_r(\lambda) = \lambda 2^{-r} \leq 2^{-K'_R} \cdot 2^{-k} \cdot \lambda \leq 2^{-K'_R} \leq \eta \, .
    \]
    Then, $\frac{p_{r+1}}{p_r} = \frac{g(x_r/2)}{g(x_r)} = h(x_r) \leq h(\eta) = q_*$, 
    and $\frac{p_{r+2}}{p_{r+1}} = \frac{g(x_{r+1}/2)}{g(x_{r+1})} = h(x_{r+1}) \leq h(\eta) \leq q_*$,
    since by Lemma~\ref{lem:hxstrictlyincreasing}, $h$ is strictly increasing and $x_{r+1} \leq x_r \leq \eta$.
    Thus, $p_{r+2}/p_r = (p_{r+2}/p_{r+1})(p_{r+1}/p_r) \leq q_*^2$, 
    so $p_{r+1}+p_{r+2} \leq (q_* + q_*^2)p_r = p_r$ by Equation~\eqref{eq:qpqsq}.
\end{proof}
\begin{corollary}
    Fix $K'_R$ as above. For any $k\geq 1$ and any $\lambda\in I_k$, the Huffman subtree induced by 
    the indices $k+K'_R, k+K'_R+1,\ldots$ is a comb with predefined merge order.
\end{corollary}
\begin{proof}
    Immediate from Lemmata~\ref{lem:combcriterion} and~\ref{lem:righttaildominance}.
\end{proof}

\subsection{Band and pool}
It is immediate that the same argument applies to the left-hand tail which decays doubly-ex\-po\-nen\-tial\-ly by Lemma~\ref{lem:ranklefttail}, 
so the left-hand tail is even more rigid and collapses into a comb structure with a simple predefined merge order,
if treated in isolation.
However, since it is also possible that some of the minuscule left-tail entries are merged with some entries in the right tail,
we need to construct a \emph{pool} of small entries that contains entries from both tails.
We will then show that the tails in the pool are rigid and cannot induce any nontrivial changes to the tree
until they have completely collapsed into a single node. To do this, we truncate the infinite right tail, 
so we can establish 
a dominance criterion between the left and the right tail: the entire left tail is dominated by every element of the
truncated right tail, so the left tail has to collapse into a single node before it can interfere with the right tail.
Similarly, the entire truncated right tail is dominated by every element of the band, so the right
tail has to collapse into a single node before it can interfere with the band.

Let us assume for now that $\lambda\in I_k$ for some $k\geq 1$, so that the mode of the distribution is at $k$ and
we can write $\lambda = \theta 2^k $ for some $\theta\in (1/2,1]$.
Let us choose a constant $0<\delta<\ln 2$, 
independent of $\lambda$, and determine an integer upper bound
on the ranks~$r_{\max}$.
\begin{lemma}
    \label{lem:rmaxwhp}
    Fix $0<\delta<\ln 2$, independent of $\lambda$, and let $r_{\max} = \lceil\log_2 \frac{\lambda}{\delta}\rceil$.
    Then, for every $k\geq 1$ and $\lambda\in I_k$, we have $\Pr[R_j > r_{\max}] \leq \delta$.
\end{lemma}
\begin{proof}
    We have 
    \begin{align*}
        \Pr[R_j > r_{\max}] = & 1 - e^{-x_{r_{\max}}} = 
    1 - e^{-\frac{\lambda}{2^{r_{\max}}}} \leq 1 - e^{-\frac{\lambda}{2^{\log_2\frac{\lambda}{\delta}}}}
    = 1 - e^{-\delta} \leq \delta \, .
    \end{align*}
\end{proof}
\begin{lemma}
    \label{lem:rmaxconstant}
    The maximum rank~$r_{\max}$ is at $O(1)$ distance from~$k$.
\end{lemma}
\begin{proof}
    We have $r_{\max} = \lceil\log_2 \frac{\lambda}{\delta}\rceil 
    = \lceil\log_2 \frac{2^k \theta}{\delta}\rceil = k + \lceil\log_2 \frac{\theta}{\delta}\rceil$.
    Therefore,
    $\log_2 \frac{1}{2\delta} < \log_2 \frac{\theta}{\delta} \leq \log_2 \frac{1}{\delta}$, 
    so $k+\log_2 \frac{1}{2\delta} < r_{\max} \leq k+\log_2 \frac{1}{\delta}$.
\end{proof}

Lemma~\ref{lem:rmaxwhp} gives us a high-probability bound on the maximum register value. Lemma~\ref{lem:rmaxconstant} 
says that the maximum register value satisfying the bound is at a constant distance from~$k$.
In practice, for fixed~$m$, we would choose $\delta$ sufficiently small that we would not expect to see any deviations
above~$r_{\max}$ in any register. The main effect $\delta$ has is that it controls the width of the band around
the mode, and what we relegate into the pool. The important thing here is that this choice does not depend on $\lambda$ or $n$ 
as such, and $r_{\max}$ is at a constant distance from~$k$. Due to the doubly exponential decay,
and the fact that the left tail is finite, 
no such truncation is needed for the left tail.

We shall now formally define the band of large weights around the mode, and the pool of small weight in the tails, with
respect to boundary constants $K_R$ and $K_L$ that are fixed later.
\begin{definition}
    \label{def:bandpool}
    Let $\mathcal L_k=\{0,1,\ldots,k-K_L-1\}$ be the set of indices in the deep left tail, possibly empty,
    and $\mathcal R_k=\{k+K_R+1, k+K_R+2, \ldots, r_{\max}\}$ be the set of indices in the deep right tail,
    possibly empty.
    The \emph{pool} is the set of indices $\mathcal P_k=\mathcal L_k\cup \mathcal R_k$.
    The \emph{band} consists of the constant-width window~$\mathcal B_k=\{k-K_L,k-K_L+1,\ldots,k+K_R-1,k+K_R\}$ around~$k$.
\end{definition}
Definition~\ref{def:bandpool} says that the band is a constant-width band around the mode, and 
the pool is the rest of the indices in the deep tails.
Furthermore, the band and the pool form a partition of the indices,
that is, $\mathcal P_k \cup \mathcal B_k = 
\mathcal L_k \cup \mathcal R_k \cup \mathcal B_k =
 \{0,1,\ldots, r_{\max}\}$ and $\mathcal P_k \cap \mathcal B_k = 
 \mathcal L_k \cap \mathcal R_k = \emptyset$.
It should be noted that either $\mathcal L_k$ or $\mathcal R_k$ can be empty,
and we have sloppily ignored the boundary cases where the indices go to zero or negative, or beyond the
upper bound on the ranks, but this has no effect on the results, as what we care about is establishing
that $|\mathcal B_k|=O(1)$, and that the pool does not affect the tree construction in a meaningful way.

We start by estimating the total masses of $\mathcal L_k$ and $\mathcal R_k$.
\begin{lemma}
    \label{lem:lefttailmass}
    The sum of probabilities in the left tail is at most 
    \[
    \sum_{\ell\in \mathcal L_k} p_\ell \leq 2e^{-2^{K_L}} \, .
    \]
\end{lemma}
\begin{proof}
    If $\ell\in\mathcal L_k$, then, by Definition~\ref{def:bandpool}, 
    $\ell\leq k-K_L-1$, so $x_{\ell} = \theta 2^{k-\ell} \geq \theta 2^{K_L+1} > 2^{K_L}$.
    Since the sum decreases faster than exponentially, it is dominated by the first term,
    so we can bound the entire sum $\sum_{\ell\in \mathcal L_k} p_\ell \leq 2e^{-2^{K_L}}$.
\end{proof}
\begin{lemma}
    \label{lem:righttailmass}
    The sum of probabilities in the right tail is at most 
    \[
        \sum_{r\in \mathcal R_k} p_r \leq 2^{-K_R} \, .
    \]
\end{lemma}
\begin{proof}
    The right tail probabilities decrease geometrically and are dominated by the first term,
    and are bounded by $2^{-K_R}$ by Lemma~\ref{lem:rankrighttail}.
\end{proof}
We need to also bound the minimum values of the right tail, and the band.
\begin{lemma}
    \label{lem:leftrighttailbandlowerbound}
    For any $k\geq 1$ and $\lambda\in I_k$, we have
    \begin{itemize}
        \item For any $r \in \mathcal R_k$, $p_r \geq g(\delta/2)$, and
        \item For any $b \in \mathcal B_k$, $p_b \geq \min\{g(2^{K_L}),g(2^{-K_R-1})\}$.
    \end{itemize}
\end{lemma}
\begin{proof}
    If $r\in \mathcal R_k$, then $r\leq r_{\max}$, and $g$ is increasing on $(0, \ln 2]$,
    so $x_r \geq x_{r_{\max}} > \delta/2$, implying $p_r = g(x_r) \geq g(\delta/2)$ for all $r\in \mathcal R_k$.

    For the second claim, if $b\in \mathcal B_k$, then by unimodality, we have that the probability is nonincreasing
    when we go to the right from $k$, and nonincreasing when we go to the left from $k$, so the minimum is attained
    at the endpoints of the band. For the left edge, we have 
    $x_{k-K_L} = \theta 2^{K_L} \leq 2^{K_L}$, so 
    since $g$ is decreasing on $[\ln 2,\infty)$, we have
    $p_{k-K_L} = g(x_{k-K_L}) \geq g(2^{K_L})$.
    For the right edge, we have $x_{k+K_R} = \theta 2^{-K_R} \geq 2^{-K_R-1}$,
    so since $g$ is increasing on $(0,\ln 2]$, we have $p_{k+K_R} = g(x_{k+K_R}) \geq g(2^{-K_R-1})$.
    Hence, $p_b \geq \min\{g(2^{K_L}),g(2^{-K_R-1})\}$ for all $b\in \mathcal B_k$.
\end{proof}

We can now establish the boundary constants $K_R$ and $K_L$.
\begin{lemma}
    \label{lem:bandpoolconstants}
    There exist constants $K_L,K_R\geq 1$, as in Definition~\ref{def:bandpool}, 
    dependent only on $\delta$ and independent of $\lambda$,
    such that 
    \begin{itemize}
    \item $\sum_{\ell \in \mathcal L_k} p_\ell < p_r$ for every $r\in \mathcal R_k$, and
    \item $\sum_{r\in \mathcal R_k} p_r < p_b$ for every $b\in \mathcal B_k$.
    \end{itemize}
\end{lemma}
\begin{proof}
    For the first claim,
    by Lemma~\ref{lem:lefttailmass}, we have $\sum_{\ell\in \mathcal L_k} p_\ell \leq 2e^{-2^{K_L}}$,
    and by Lemma~\ref{lem:leftrighttailbandlowerbound}, we have $p_r \geq g(\delta/2)$ for every $r\in \mathcal R_k$.
    Thus, it suffices to choose $K_L$ sufficiently large such that $2e^{-2^{K_L}} < g(\delta/2)$.
    We can do this since $g(\delta/2)$ is a positive constant independent of $K_L$, 
    and $\lim_{K_L \to \infty} 2e^{-2^{K_L}} = 0$.

    For the second claim, by Lemma~\ref{lem:righttailmass}, we have $\sum_{r\in \mathcal R_k} p_r \leq 2^{-K_R}$,
    and by Lemma~\ref{lem:leftrighttailbandlowerbound}, we have $p_b \geq \min\{g(2^{K_L}),g(2^{-K_R-1})\}$ 
    for every $b\in \mathcal B_k$.
    Thus, it suffices to choose $K_R$ sufficiently large such that $2^{-K_R} < \min\{g(2^{K_L}),g(2^{-K_R-1})\}$.
    We can do this since $g(2^{K_L})$ and $g(2^{-K_R-1})$ are positive constants independent of $K_R$, 
    and $\lim_{K_R \to \infty} g(2^{-K_R-1}) = 0$.
\end{proof}
\begin{corollary}
    Lemma~\ref{lem:bandpoolconstants} implies that the band and the pool are well-defined, 
    and the band has constant width, independent of $\lambda$.
\end{corollary}
\begin{corollary}
    \label{cor:bandpoolcomb}
    Lemma~\ref{lem:bandpoolconstants} implies that, applying the Huffman algorithm on the distribution,
    the left-hand tail first collapses into a single node, and then the right-hand tail collapses into a single node.
    Both the left-hand tail and the right-hand tail satisfy the comb criterion in Lemma~\ref{lem:combcriterion} 
    and collapse into a comb structure with a predefined merge order. Thus, the only nontrivial merges that can occur
    are those involving the band and only the band.
\end{corollary}
\begin{remark}
    We tacitly ignored the fact that the boundary constants in Lemma~\ref{lem:bandpoolconstants} can be rather large, and the band can extend beyond the lower
    or upper bounds of the ranks, but this is of no consequence, as we only care about bounding the band to be of constant width
    asymptotically. The distribution is rather well-behaved in practice, though.
\end{remark}

\subsection{Tree changes on an interval}
Corollary~\ref{cor:bandpoolcomb} says that the only nontrivial merges that can occur are those involving the constant-width band
and only the band, as the tails collapse into combs before interacting with the band. However, this is not enough
to establish the number of tree changes on an interval $I_k$. To do so, we establish that all merges involve
a subset of leaves, so their weights are finite sums of exponentials. Observing that the merge order can only change
if a difference between two weights changes sign, we can bound the number of tree changes by bounding the number of
zeros on such sums, which is constant independent of $\lambda$ by Rolle's theorem, thus establishing the claim.

We need the following two lemmata, the latter of which presents a bound on the number of zeros of bounded-term sums of exponentials, which is a consequence of 
Rolle's theorem.
\begin{lemma}
    \label{lem:exponentialtelescoping}
    For integers $1 \leq a \leq b$,
    \[
    \sum_{r=a}^b p_r(\lambda) = e^{-\frac{\lambda}{2^b}} - e^{-\frac{\lambda}{2^a}} \, .
    \]
\end{lemma}
\begin{proof}
    Follows directly from the properties of the distribution as discussed in Section~\ref{sec:distribution}.
\end{proof}
\begin{corollary}
    \label{cor:tailzeros}
    The nodes that represent the collapsed left and right tails have weights that are a sum of at most four exponentials.
\end{corollary}
\begin{proof}
    The claim follows by observing that both the left and right tails consist of consequtive indices and by 
    Lemma~\ref{lem:exponentialtelescoping}.
\end{proof}
\begin{lemma}
    \label{lem:rollezeros}
    Let $0 < a_1 < a_2 < \cdots < a_\ell$, and
    \[
    F(\lambda) = \sum_{i=1}^\ell c_i e^{-a_i\lambda} \, ,
    \]
    such that not all $c_i$ are zero. Then, $F$ has at most $\ell-1$ distinct zeros on $(0,\infty)$.
\end{lemma}
\begin{proof}
    We induct on $\ell$. The base case $\ell=1$ is trivial, as $F$ has no zeros, by assumption that $c_1 \neq 0$.

    Assuming the claim holds for $\ell-1$, consider $G(\lambda) = e^{a_1\lambda} F(\lambda)$.
    Clearly, $G$ and $F$ have the same zeros, so it suffices to show that $G$ has at most $\ell-1$ distinct zeros.
    By Rolle's theorem, between any two distinct zeros of $G$ there is a zero of $G'$.
    We have $G'(\lambda) = \sum_{i=2}^\ell c_i (a_1 - a_i) e^{-(a_i - a_1)\lambda}$, 
    which is also a sum of $\ell-1$ exponentials with at most $\ell-2$ distinct exponents by the induction hypothesis.
    Thus, the claim holds for $G$.
\end{proof}
\begin{remark}
The result of Lemma~\ref{lem:rollezeros} could also have been obtained by arguing that $F$ is an analytic function
on a compact interval. 
\end{remark}
We can now argue that there are only a constant number of tree changes
as $\lambda$ varies over the interval $I_k$.
\begin{proposition}
    \label{prop:treechangesoninterval}
    For fixed $k\geq 1$, as $\lambda$ varies over $I_k$, 
    the Huffman tree can change at most $O(1)$ times.
\end{proposition}
\begin{proof}
    Since the tails have a fixed merge order by Corollary~\ref{cor:bandpoolcomb}, 
    the only way the structure of the tree can change is if there is a tie between two candidate merges involving the band 
    that can change the structure, that is, the difference between the weights of the two candidate merges changes sign 
    as $\lambda$ varies.

    Since $|\mathcal B_k|=O(1)$, there are only $O(1)$ candidate merges that can involve the band, and 
    hence only $O(1)$ candidate merges that can change the structure. Furthermore, since the weights of the nodes are
    finite-term sums of exponentials, as they are the sums of singleton weights for some subset of $\mathcal B_k$,
    and there are only constantly many such subsets, plus
    the tail node by Corollary~\ref{cor:tailzeros}, 
    the difference between the weights of any two candidate merges is a finite-term sum of exponentials, 
    so by Lemma~\ref{lem:rollezeros}, there can only be $O(1)$ values of $\lambda$ that can induce a tree change.
\end{proof}

\subsection{Total number of tree changes}
We now bound the total number of tree changes, that is, when when $\lambda$ goes up to $\frac{N}{m}$ for some upper bound~$N$ on the cardinality~$n$.
Let us recall the statement of Proposition~\ref{prp:huffmantreechanges} and prove it.
\prphuffmantreechanges*
\begin{proof}
    If $\lambda\leq 1$, then $[0,\lambda]\in I_0$, meaning $r^*(\lambda) = 1$, which implies there is no left tail.
    Thus, the right tail decays into a comb, yielding at most $O(1)$ tree changes.
    Consequently, if $\lambda_N \leq 1$, the total number of tree changes is $O(1)$, so assume $\lambda_N > 1$.

    The interval $[0,\lambda_N]$ intersects at most $O(\log N)$ intervals of the form $I_k$, and by Proposition~\ref{prop:treechangesoninterval}, 
    there are at most $O(1)$ tree changes on each such interval, so the total number of tree changes is $O(\log N)$.
\end{proof}

\section{Proofs of Theorems~\ref{thm:amortizedupdates} and~\ref{thm:amortizedupdateswithconstantwords}}
\label{app:amortizedupdatesproofs}

We will now analyze the
remaining update types. Throughout this section, we assume that $N$ is the total number of updates, an upperbound on the cardinality, and
$n=1,2,\ldots,N$ is the current cardinality at a certain time step. We fix~$m$, and set $\lambda = \frac{n}{m}$.
After $n$ updates, we receive the $(n+1)^{\text{th}}$ element with a new rank~$R\sim\textrm{Geom}\left(\frac 12\right)$ 
that is independent of $R_{\min}$ or any current register values. We denote as usual $x_r = \frac{\lambda}{2^r}$, $p_0=e^{-x_0}$ and $p_r = e^{-x_r}(1-e^{-x_r})$ for $r\geq 1$.

Let $R_{\min}=\min_{j\in[B]} R_j$ be the minimum rank contained in a bucket. 
The following lemma ais necessary for analyzing the updates.
\begin{lemma}
    \label{lem:bucketmin}
    For any integer $r\geq 0$, we have 
    \[
    \Pr[R_{\min}\leq r] = 1 - \Pr[R_j>r]^B = 1 - (1-e^{-x_r})^B \, ,
    \]
    and
    \[
    \Pr[R_{\min} = r] = \Pr[R_{\min}\leq r] - \Pr[R_{\min}\leq r-1] = (1-e^{-x_{r-1}})^B - (1-e^{-x_r})^B \, .
    \]
\end{lemma}
\begin{proof}
    By definition: the minimum rank can be at most $r$ if and only if it is not the case that all registers are greater than $r$.
\end{proof}

An ordinary update is triggered if and only if $R\geq R_{\min}$, otherwise we can just discard the element.
The probability of such an update is by independence of $R$ and $R_{\min}$, Lemma~\ref{lem:bucketmin}, 
and marginalizing over $r$,
\begin{equation}
    \label{eq:ordinaryupdateprobability}
    \begin{split}
    \Pr[R\geq R_{\min}] & = \sum_{r=0}^\infty \sum_{r'=r}^\infty \Pr[R_{\min} = r] \Pr[R = r'] \\
    & = \sum_{r=0}^\infty \Pr[R_{\min} = r] \sum_{r'=r}^\infty \Pr[R = r'] \\
    & = \Pr[R_{\min} = 0] + \sum_{r=1}^\infty \Pr[R_{\min} = r] \Pr[R \geq r] \\
    & = \Pr[R_{\min} = 0] + \sum_{r=1}^\infty \Pr[R_{\min} = r] 2^{-r-1} \\
    & = \left(1-(1-e^{-\lambda})^B\right) + \sum_{r=1}^\infty \left((1-e^{-x_{r-1}})^B - (1-e^{-x_r})^B\right) 2^{-r-1} 
    \, .
    \end{split}
\end{equation}

Let us now provide an upper bound on the number of ordinary updates.
\begin{proposition}
    \label{prop:ordinaryupdates}
    Assuming $m=O(N/\log^3 N)$, after $N$ insertions, the expected number of ordinary updates is at most $O(N/\log N)$.
\end{proposition}
\begin{proof}
    For simplicity, let us write $\lambda_n=\frac{n}{m}$ for $n=1,2,\ldots,N$,
    $f(x) = (1-e^{-x})^B$, 
    and rewrite Equation~\eqref{eq:ordinaryupdateprobability} as
    \begin{equation}
        \label{eq:ordinaryupdateprobability2}
    S(\lambda) = f(\lambda) + \sum_{r=1}^\infty \left(f(x_{r-1}) - f(x_r)\right) 2^{-r-1} \, .
    \end{equation}
    We need to evaluate the sum
    \begin{equation}
        \label{eq:ordinaryupdateprobabilitysum}
        \sum_{n=0}^N S(\lambda_n) \, .
    \end{equation}
    To do this, we shall bound the inner sum of Equation~\eqref{eq:ordinaryupdateprobability2}. Let us split
    Equation~\eqref{eq:ordinaryupdateprobability2} into three parts: the constant term $f(\lambda)$, 
    the first ranks $1 \leq r < r^*=\lceil \log_2 \lambda \rceil$,
    and the remaining terms $r \geq r^*$:
    \begin{equation}
        \label{eq:ordinaryupdateprobability3}
    S(\lambda) = f(\lambda) + \sum_{r=1}^{r^*-1} \left(f(x_{r-1}) - f(x_r)\right) 2^{-r-1} + \sum_{r=r^*}^\infty \left(f(x_{r-1}) - f(x_r)\right) 2^{-r-1} \, .
    \end{equation}

    The first term of Equation~\eqref{eq:ordinaryupdateprobability3} bounded by Bernoulli's inequality as
    \begin{equation}
        \label{eq:constanttermbound}
        1 - (1-e^{-\lambda})^B \leq B e^{-\lambda} \leq \frac{B}{\lambda}
    \end{equation}

    For the second term, observe that $f$ is increasing and $f(x_{r-1}) \leq 1$, so therefore
    \[
    f(x_{r-1}) - f(x_r) \leq 1 - f(x_r) = 1 - (1-e^{-x_r})^B \leq Be^{-x_r} \, ,
    \]
    hence
    \begin{equation}
        \label{eq:firstsumbound}
        \sum_{r=1}^{r^*-1} \left(f(x_{r-1}) - f(x_r)\right) 2^{-r-1} 
        \leq \frac{B}{2} \sum_{r=1}^{r^*-1} e^{-x_r} 2^{-r}
        = \frac{B}{2} \sum_{r=1}^{r^*-1} e^{-\frac{\lambda}{2^r}} 2^{-r} \, .
    \end{equation}
    Substituting $k=r^*-r$, we have $2^{-r}e^{-\frac{\lambda}{2^r}} = 2^{-r^*}2^k e^{-\frac{\lambda}{2^r}}$. 
    Furthermore, we observe that $\frac{1}{2}<\frac{\lambda}{2^{r^*}} \leq 1$, implying
    $e^{-\frac{\lambda}{2^r}} \leq e^{-2^{k-1}}$, so
    we can further bound the sum of Equation~\eqref{eq:firstsumbound} as
    \begin{equation}
        \label{eq:firstsumbound2}
        \sum_{r=1}^{r^*-1} \left(f(x_{r-1}) - f(x_r)\right) 2^{-r-1} 
        \leq \frac{B}{2} 2^{-r^*} \sum_{k=1}^{\infty} e^{-2^{k-1}} 2^k = O\left(\frac{B}{\lambda}\right)
        \, ,
    \end{equation}
    by the fact that the series converges to a constant and $r^*=\lceil \log_2 \lambda \rceil$.
    
    Finally, for the third term of Equation~\eqref{eq:ordinaryupdateprobability3}, we have
    $x_r \leq 1$ for $r\geq r^*$, so therefore $1-e^{x_r} \leq x_r$, and hence
    $f(x_r) = (1-e^{-x_r})^B \leq x_r^B$ for $r \geq r^*$, so
    $f(x_{r-1}) - f(x_r) \leq f(x_{r-1}) \leq x_{r-1}^B$.
    Therefore, 
    \[
    \sum_{r=r^*+1}^\infty \left(f(x_{r-1}) - f(x_r)\right) 2^{-r-1} \leq
    \sum_{r=r^*+1}^\infty \left(\frac{\lambda}{2^{r-1}}\right)^B 2^{-r-1} = O\left(\frac{1}{\lambda}\right)\,,
    \]
    because $2^{r^*-1} < \lambda \leq 2^{r^*}$, so the series is $O(2^{-r^*}) = O\left(\frac{1}{\lambda}\right)$.
    
    We need to check the special case of $r=r^*$, where we have 
    $0\leq \left(f(x_{r^*-1}) - f(x_{r^*})\right)2^{-r^*-1} \leq 2^{-r^*-1} = O\left(\frac{1}{\lambda}\right)$.

    Putting this all together, we get that the inner sum of Equation~\eqref{eq:ordinaryupdateprobability3} is at 
    most $S(\lambda) = O\left(\frac{B}{\lambda}\right)=O\left(\frac{Bm}{n}\right)$, so we get
    \[
        \sum_{n=0}^N S(\lambda_n) = O\left(Bm \log N\right) = O\left(m \log^2 N\right)\, ,
    \]
    by approximating the harmonic series with the natural logarithm, and considering that $B=O(\log N)$.
    Finally, if $m=O(N/\log^3 N)$, then the expected number of ordinary updates is at most $O(N/\log N)$, thus
    yielding the claim.
\end{proof}
\begin{corollary}
    \label{cor:amortizedordinaryupdates}
    Assuming an ordinary update takes $O(\log n)$ time and that $m=O(n/\log^3 n)$, the amortized cost of an ordinary update is O(1)
    over $n$ updates.
\end{corollary}
Corollary~\ref{cor:amortizedordinaryupdates} makes intuitive sense: the expected number of ordinary updates is only bounded to a reasonable
number if the number of elements is sufficiently large compared to the number of registers, as otherwise almost all registers are empty,
and we are better off using linear counting~\cite{WhangVT:1990}.

Next, we analyze the number of minimum-rank updates.
\begin{proposition}
    \label{prop:minimumrankupdates}
    The expected number of minimum-rank updates after $N$ insertions is at most $O(m/\log N)$.
\end{proposition}
\begin{proof}
    There are $O\left(\frac{m}{B}\right) = O\left(\frac{m}{\log N}\right)$ buckets,
    and there are $O(\log N)$ possible rank values.
    Therefore, there are at most $O(m)$ possible minimum-rank updates.
\end{proof}
\begin{corollary}
    \label{cor:amortizedminimumrankupdates}
    Assuming a minimum-rank update takes $O(\log^2 n)$ time and that $m=O(n/\log^3 n)$, the amortized cost of a minimum-rank update is O(1)
    over $n$ updates.
\end{corollary}
Corollary~\ref{cor:amortizedminimumrankupdates} is rather loose, but since this is subsumed by the cost of ordinary updates, it suffices for our 
purposes.
We can now put these results together to obtain the following results.
Let us recall the statement of Theorem~\ref{thm:amortizedupdates} and prove it.
\thmamortizedupdates*
\begin{proof}
    Since we want to bound the amortized costs over $n$ updates, we can just bound all of these separately and sum the bounds.
    By Proposition~\ref{prp:huffmantreechanges}, there are at most $O(\log n)$ tree updates, 
    so the total cost of tree updates is at most $O(m\log^2 n)$, which is $O(n/\log n)$ under the assumption on $m$.
    The total cost of ordinary updates is at most $O(n)$ by Proposition~\ref{prop:ordinaryupdates}, 
    and the total cost of minimum-rank updates is at most $O(m\log n)$ by Proposition~\ref{prop:minimumrankupdates}, 
    which is $O(n/\log^2 n)$ under the assumption on $m$.
    Thus, the total cost of all updates is at most $O(n/\log n + n + n/\log^2 n) = O(n)$, 
    so the amortized cost is O(1), as claimed.
\end{proof}

Let us recall the statement of Theorem~\ref{thm:amortizedupdateswithconstantwords} and prove it.
\thmamortizedupdateswithconstantwords*
\begin{proof}
    By Proposition~\ref{prp:huffmantreechanges}, there are at most $O(\log n)$ tree updates, 
    so the total cost of tree updates is at most $O(m\log n)$, which is $O(n/\log n)$ under the assumption on $m$.
    The total cost of ordinary updates is at most $O(n)$ by Proposition~\ref{prop:ordinaryupdates} if we adjust
    for $m$, 
    and the total cost of minimum-rank updates is at most $O(m)$ by Proposition~\ref{prop:minimumrankupdates}, 
    which is $O(n/\log^2 n)$ under the assumption on $m$.
    Thus, the total cost of all updates is at most $O(n)$, so the amortized cost is O(1), as claimed.
\end{proof}

\section{Pseudocode for the operations}
\label{app:pseudocode}
Pseudocode for the operations of the HBS sketch is given in Algorithm~\ref{alg:hbs}.

\begin{algorithm}[t]
    \begin{algorithmic}[1]
        \Function{Peek}{$b,j$}
            \State Determine the $j^\textrm{th}$ codeword start~$s_j$ and length~$\ell_j$ in bucket~$b$ with the unary array
            \State Get the codeword $c_j$ from the bits $s_j, s_j+1, \ldots, s_j+\ell_j-1$ in the codeword array
            \State Decode $c_j$ using the Huffman tree or lookup table to obtain the rank value~$r$
            \State \Return $r$
        \EndFunction
        \Function{Poke}{$r,b,j$}
            \State Encode $r$ into a codeword $c$ using the Huffman tree or lookup table
            \State Replace $c_{\textrm{old}}$ with $c$, shifting the bucket if necessary
        \EndFunction
        \Function{Insert}{$y$}
        \State Compute $b\gets h_b(y)$, $j\gets h_j(y)$, and $r\gets \rho(h_r(y))$.
        \If{$r>r_{\min}$ of the bucket~$b$}
            \State $r_{\mathrm{old}} \gets \textsc{Peek}(b,j)$
            \If{$r>r_{\mathrm{old}}$}
                \State \Call{Poke}{$r,b,j$}
                \If{$r_{\mathrm{old}}=r_{\min}$}
                    \State $c_{\min} \gets c_{\min}-1$
                    \If{$c_{\min}=0$}
                        \State Recompute $r_{\min}$ and $c_{\min}$ by peeking all registers in the bucket
                    \EndIf
                \EndIf
                \State Update $\hat{n}_b$ and $\hat{n}$.
                \If{$\hat{n}-\hat{n}_{\textrm{old}}$ is too large}
                    \State Reconstruct the Huffman tree and re-encode all buckets
                    \State $\hat{n}_{\textrm{old}} \gets \hat{n}$
                \EndIf
            \EndIf
        \EndIf
        \EndFunction
        \Function{Merge}{$S_1,S_2$}
            \ForAll{$b\in[m/B]$}
                \State Construct $\hat{n}_b$ by iterating over $j\in[B]$ and taking max of the $r_j$ from the sketches
            \EndFor
            \State $\hat{n} \gets \sum \hat{n}_b$
            \If{$\hat{n}-\max\{\hat{n}_{\textrm{old},1},\hat{n}_{\textrm{old},2}\}$ is too large}
                \State Construct the Huffman tree using $\hat{n}$
            \Else
                \State Reuse the Huffman tree from the sketch with the larger $\hat{n}_{\textrm{old}}$
            \EndIf
            \ForAll{$(b,j)\in[m/B]\times [B]$}
                \State $r\gets \max\{\textsc{Peek}_{S_1}(b,j), \textsc{Peek}_{S_2}(b,j)\}$
                \State \Call{Poke}{$r,b,j$}
            \EndFor
        \EndFunction
    \end{algorithmic}
    \caption{Huffman-Bucket Sketch operations.}
    \label{alg:hbs}
\end{algorithm}

\end{document}